\documentclass[11pt,a4paper]{article}
\usepackage[T1]{fontenc}
\usepackage[utf8]{inputenc}
\usepackage{lmodern}
\usepackage[margin=28mm]{geometry}
\usepackage{amsmath,amssymb,amsthm,mathtools}
\usepackage{microtype,booktabs,array}
\usepackage[shortlabels]{enumitem}
\usepackage{tikz}
\usetikzlibrary{arrows.meta,positioning,calc}
\usepackage[colorlinks=true,linkcolor=blue!50!black,citecolor=blue!50!black,urlcolor=blue!50!black]{hyperref}
\usepackage[nameinlink,capitalize,noabbrev]{cleveref}
\usepackage{aliascnt}
\newtheorem{theorem}{Theorem}[section]
\newcommand{\sharedtheorem}[3]{%
  \newaliascnt{#1}{theorem}%
  \newtheorem{#1}[#1]{#2}%
  \aliascntresetthe{#1}%
  \crefname{#1}{#2}{#3}\Crefname{#1}{#2}{#3}}
\sharedtheorem{lemma}{Lemma}{Lemmas}
\sharedtheorem{fact}{Fact}{Facts}
\sharedtheorem{proposition}{Proposition}{Propositions}
\sharedtheorem{corollary}{Corollary}{Corollaries}
\theoremstyle{definition}
\sharedtheorem{definition}{Definition}{Definitions}
\sharedtheorem{example}{Example}{Examples}
\sharedtheorem{remark}{Remark}{Remarks}
\newtheorem*{restatedstructure}{Theorem~\ref*{thm:structure}}
\newtheorem*{restatedpartners}{Lemma~\ref*{lem:partners}}
\newtheorem*{restatedresidualdegree}{Lemma~\ref*{lem:residual-degree}}

\newcommand{\eps}{\varepsilon}
\newcommand{\Cout}{C_{\mathrm{out}}}
\newcommand{\LOCAL}{\mathsf{LOCAL}}

\newcommand{\hard}{\widehat D_6}

\hypersetup{
  pdftitle={Near-Optimal Distributed Domination in Planar Graphs},
  pdfauthor={Wojciech Wawrzyniak},
  pdfkeywords={Distributed algorithms, minimum dominating set, planar graphs, LOCAL model}}
\title{Near-Optimal Distributed Domination\\in Planar Graphs}
\author{Wojciech Wawrzyniak\\[2pt]
  \normalsize Faculty of Mathematics and Computer Science\\
  \normalsize Adam Mickiewicz University, Pozna\'n, Poland\\
  \normalsize\texttt{wwawrzy@amu.edu.pl}}
\date{}
\begin{document}
\maketitle
\begin{abstract}
We give a deterministic $(8+\eps)$-approximation for minimum dominating
set on planar graphs in a constant number of rounds of the $\LOCAL$
model, for every $\eps>0$. This improves the previous ratio $11+\eps$
obtained by Heydt et al. The ratio is near-optimal in this model:
its leading constant is only one above the known lower bound of $7$.
Our result closes three quarters of the previous gap, reducing it from $4$ to $1$.
Our main contribution is a sharp structural bound.
For any dominating set $D$, assigning each vertex outside $D$ to a
neighboring center gives disjoint owner blocks.
If $k_x$ counts the other blocks containing a neighbor of $x$, then
$\sum_{x\notin D}(k_x-2)^+\le(4|D|-12)^+$, where $z^+=\max\{z,0\}$.
The bound holds for every such assignment, and equality holds for
arbitrarily large minimum dominating sets.
We use this bound in their three-phase framework, with
new parameters and the same final linear-programming procedure.
The algorithm requires neither a planar embedding nor the graph size,
and its round bound depends only on $\eps$.
The transfer theorem of Bonamy et al. also gives a deterministic
$(25+\eps)$-approximation on graphs of bounded Euler genus, with a
round bound depending only on $\eps$ and the genus.
\end{abstract}
\section{Introduction}
\label{sec:introduction}

A dominating set of a graph $G=(V,E)$ is a set $D\subseteq V$ such that
every vertex is in $D$ or has a neighbor in $D$. The minimum dominating
set problem asks for such a set of smallest size. An $\alpha$-approximation
returns a dominating set of size at most $\alpha$ times this minimum.

In a distributed network, a dominating set gives a way to place
local coordinators. We study this problem in the $\LOCAL$ model, where
vertices with unique identifiers communicate in synchronous rounds,
with no limit on message size or local computation
\cite{Linial1992,NaorStockmeyer1995}. In a fixed number
of rounds, a vertex can learn only about the part of the graph within
a fixed distance from it. For a survey of local algorithms, see
Suomela~\cite{Suomela2013Survey}. On planar graphs, there are
constant-factor approximations in a constant number of communication
rounds. However, the best possible approximation factor is still not known.

Lower bounds show what a constant-round algorithm cannot achieve.
In 2008, Czygrinow et al.
\cite{CzygrinowHanckowiakWawrzyniak2008} used Ramsey's theorem to prove
that, for every fixed
$\eps>0$, no deterministic constant-round algorithm can find a
$(5-\eps)$-approximation of minimum dominating set on all planar graphs.
In 2014, Hilke et al.~\cite{HilkeLenzenSuomela2014} adapted this approach
to triangular grids and improved the lower bound to $7-\eps$.

On the algorithmic side, Lenzen et al. gave a
constant-round algorithm with approximation factor 130
\cite{LenzenPignoletWattenhofer2013}. Wawrzyniak showed that the same
algorithm has approximation factor 52 \cite{Wawrzyniak2014}.
Heydt, Siebertz and Vigny gave a factor-20 algorithm in 2022
\cite{HeydtSiebertzVigny2022}. More recently, their 2025 journal paper
with Kublenz and Ossona de Mendez presented a factor-$(11+\eps)$
three-phase algorithm whose final phase uses linear programming
\cite{HeydtEtAl2025}.
Here we use this framework with different parameters and a new
structural bound to obtain factor $8+\eps$.

\begin{theorem}
\label{thm:main}
For every $\eps>0$, there is a deterministic algorithm that computes
an $(8+\eps)$-approximation of minimum dominating set on every
simple planar graph in a constant number of rounds of the $\LOCAL$ model.
\end{theorem}

Our approximation ratio is near-optimal for deterministic constant-round
algorithms. The lower bound of Hilke et al. rules out every factor below seven.
Our result reduces the gap between the leading approximation constant
and this lower bound from $4$ to $1$, closing three quarters of that gap.
Whether factor $7+\eps$ is
achievable on all planar graphs remains open.

We group the vertices into disjoint stars centered at a dominating set $D$.
Our main contribution is a structural bound on contacts from individual
vertices to other stars. It holds for every dominating set and every
assignment to neighboring centers, and its leading coefficient is sharp.
The stars are used only in the proof. The algorithm does not need to
find them.

\paragraph*{Comparison with Heydt et al.}
We use the three phases of Heydt et al.~\cite{HeydtEtAl2025},
including the same linear-programming procedure in the final phase.
One main change is in the first phase. Both the factor-$20$ and
factor-$(11+\eps)$ algorithms of Heydt et al.
\cite{HeydtSiebertzVigny2022,HeydtEtAl2025} select every vertex whose
open neighborhood cannot be dominated by at most $3$ other vertices.
We use threshold $6$, as in \cite{LenzenPignoletWattenhofer2013}.
Our structural theorem gives a stronger bound on the size of the selected set.

In their proof, Heydt et al. first choose an independent subset of the vertices that
satisfy their neighborhood-domination condition. They then contract
stars and use a bound on the number of edges
\cite{HeydtEtAl2025}.

One edge after contraction can represent contacts from many original
vertices. We count these contacts before contraction, without choosing
an independent set. In a triangulation with a 2-connected contact graph,
we use corners of triangular faces that meet three blocks, leaving two
corners per block unused. We then split at cut vertices and bound the
extra weight caused by the split. This gives the bound for all planar graphs.

Recently, in 2026, Bonamy et al.~\cite{BonamyEtAl2026} proved a transfer
theorem. With our planar result, it gives a deterministic
$(25+\eps)$-approximation on graphs of bounded Euler
genus. This improves their ratio $34+\eps$, which is also independent
of the genus. For Euler genus $1\le g\le5$, the bound of Heydt et al.
is smaller than $25+\eps$. See \cref{cor:genus} for details.

\paragraph*{Related work.}
On general graphs, constant rounds do not suffice for a constant
approximation factor \cite{KuhnMoscibrodaWattenhofer2016}.
Smaller constants are known for planar subclasses. In 2025, Wawrzyniak
gave a factor-$6$ algorithm for triangle-free planar graphs
\cite{Wawrzyniak2025TriangleFree}. This also applies to bipartite planar
graphs and planar graphs of girth at least five. For outerplanar graphs,
the optimal factor is $5$ \cite{BonamyEtAl2021Outerplanar}.

Beyond planar graphs, Amiri, Schmid and Siebertz gave a constant-factor approximation in a
constant number of rounds on graphs of bounded genus
\cite{AmiriSchmidSiebertz2016}. Their journal paper gives such an
algorithm for locally embeddable graph classes and a $(1+\eps)$-approximation
in $O(\log^* n)$ rounds on graphs of bounded genus
\cite{AmiriSchmidSiebertz2019}.
Constant-round constant-factor approximations also exist for classes
excluding a fixed topological minor \cite{CzygrinowEtAl2018} and for
classes of bounded expansion \cite{KublenzSiebertzVigny2021,HeydtEtAl2025}.
For every fixed $\eps>0$, a $(1+\eps)$-approximation can be computed in
$O(\log^* n)$ rounds on planar and excluded-minor classes
\cite{CzygrinowHanckowiakWawrzyniak2008,CzygrinowEtAl2018}.

For distance-$k$ domination on graphs excluding $K_{2,t}$ as a minor,
Czygrinow, Han\'{c}kowiak and Witkowski
\cite{CzygrinowHanckowiakWitkowski2022} give a constant-factor
approximation in $O(k)$ rounds and a $(1+\eps)$-approximation in
$O(\log^* n)$ rounds, for fixed $k,t$ and $\eps>0$.

Results are also known under message-size restrictions.
CONGEST allows $O(\log n)$ bits per edge per round. CONGEST-BC also
requires the same message to all neighbors. Constant-factor approximations
in constant rounds exist for anonymous planar networks with short messages
\cite{Wawrzyniak2015Anonymous} and for
bounded-genus graphs in CONGEST-BC \cite{CzygrinowEtAl2019Broadcast}.
For bounded expansion, Amiri et al.~\cite{AmiriEtAl2018Expansion} give
a constant-factor approximation in $O(\log n)$ rounds in CONGEST-BC.
Deurer et al.~\cite{DeurerKuhnMaus2019} give deterministic CONGEST
algorithms for general graphs with factors logarithmic in the maximum
degree and nonconstant round bounds. Our LOCAL model has no message-size limit.

\paragraph*{Why the structural bound gives eight.}
Let $D$ be an optimum dominating set and let $D_1,D_2,D_3$ be the sets
selected in the three phases of the algorithm.
Our bound gives $|D_1\cup D_2|\le |D|+7|(D_1\cup D_2)\cap D|$.
The third phase gives
$|D_3|\le(7+\eps)(|D|-|(D_1\cup D_2)\cap D|)$.
Together, these bounds give
$|D_1\cup D_2\cup D_3|\le(8+\eps)|D|$.

We prove the structural bound in \cref{sec:triangulations,sec:cuts}
and analyze the algorithm in \cref{sec:algorithm,sec:residual}.

\section{Owner blocks and the structural theorem}
\label{sec:structural}

All graphs are finite, simple and undirected, except where we say that
we use a multigraph. We write $uv$ for the edge $\{u,v\}$,
$N(v)$ for the open neighborhood, $N[v]=N(v)\cup\{v\}$, and
$N[Z]=\bigcup_{z\in Z}N[z]$. For a set $X$, write $X\setminus v$
for $X\setminus\{v\}$. A \emph{plane graph} is a planar graph
with a fixed drawing without crossings. We work on the sphere, so
there is no need to choose an outer face. For a \emph{planar graph},
such a drawing exists, but no drawing is fixed.
For $z\in\mathbb R$, we write $z^+=\max\{z,0\}$.

\begin{definition}
\label{def:owners}
Let $D$ be a dominating set. Its vertices are called \textbf{centers}.
All other vertices are \textbf{non-centers}. An \textbf{owner assignment}
is a map $m:V(G)\to D$ such that $m(d)=d$ for $d\in D$, and
$xm(x)\in E(G)$ for $x\notin D$. Thus each center owns itself, and
each non-center is assigned to one of its neighbors in $D$.
\end{definition}

\begin{definition}
\label{def:owner-blocks}
Fix a dominating set $D$ and an owner assignment $m$.
For $d\in D$, the \textbf{owner block} $B_d=m^{-1}(d)$ is the set of vertices assigned
to $d$. The edges $dx$, for $x\in B_d\setminus d$, form a star
with center $d$ that spans $B_d$. We call this star $T_d$ the
\textbf{owner star} of $d$.
\end{definition}

\begin{definition}
\label{def:foreign-contacts}
For $x\notin D$, a block is \textbf{foreign} if it is not $B_{m(x)}$.
A \textbf{foreign contact} means that $x$ has a neighbor in a foreign block.
This neighbor does not have to be the center of that block.
Define the set of centers of the blocks with which $x$ has such a
contact, and its size, by
\[
 \Cout(x)=\{c\in D\setminus m(x):N(x)\cap B_c\ne\varnothing\},
 \qquad k_x=|\Cout(x)|.
\]
\end{definition}

Thus $k_x$ counts foreign blocks, not edges or centers adjacent to $x$.
Several neighbors in the same foreign block count only once.
See graph $G$ in \cref{fig:owners}.

\begin{definition}
\label{def:contact-weight}
Define the \textbf{total weight} over all non-centers by
\[
 W=\sum_{x\in V(G)\setminus D}(k_x-2)^+.
\]
The \textbf{weight} of a non-center $x$ is $(k_x-2)^+$. It is positive exactly
when $k_x\ge3$.
\end{definition}

Extra edges inside a block are allowed, so a block need not be just a star.
Centers are not
included in $W$, even when they have neighbors in many foreign blocks.

\begin{figure}[!htbp]
\centering
\begin{tikzpicture}[scale=.9, every node/.style={font=\small}]
\tikzset{center/.style={circle,draw,fill=blue!15,minimum size=7mm},
  vertex/.style={circle,draw,fill=white,minimum size=6mm,inner sep=0pt},
  box/.style={rounded corners,draw=blue!60!black,fill=blue!3}}
\draw[box] (-.6,-1.1) rectangle (2.5,1.1);
\node[center] (d) at (0,0) {$d$};
\node[vertex] (x) at (1.6,0) {$x$};
\draw[thick,-{Stealth}] (x)-- node[above=3mm] {$m(x)=d$} (d);
\node at (.9,-.75) {$B_d$};
\foreach \i/\yy in {1/1.35}{
  \draw[box] (4.1,\yy-.55) rectangle (6.7,\yy+.55);
  \node[vertex] (y\i) at (4.65,\yy) {$y_\i$};
  \node[center] (c\i) at (6.1,\yy) {$c_\i$};
  \draw[thick,-{Stealth}] (y\i)--(c\i);
  \draw[dashed,red!65!black,thick] (x)--(y\i);
}
\draw[box] (5.5,-.55) rectangle (6.7,.55);
\node[center] (c2) at (6.1,0) {$c_2$};
\draw[dashed,red!65!black,thick] (x)--(c2);
\draw[box] (4.1,-2.55) rectangle (6.7,-.65);
\node[vertex] (y3) at (4.65,-1.1) {$y_3$};
\node[vertex] (y3prime) at (4.65,-2.1) {$y_3'$};
\node[center] (c3) at (6.1,-1.6) {$c_3$};
\draw[thick,-{Stealth}] (y3)--(c3);
\draw[thick,-{Stealth}] (y3prime)--(c3);
\draw[thick] (y3)--(y3prime);
\draw[dashed,red!65!black,thick] (x)--(y3);
\draw[dashed,red!65!black,thick] (x)--(y3prime);
\node at (-1.15,0) {$G$};
\node[center] (hd) at (9,0) {$d$};
\foreach \i/\yy in {1/1.35,2/0,3/-1.6}{
  \node[center] (hc\i) at (11.7,\yy) {$c_\i$};
  \draw[thick] (hd)--(hc\i);
}
\node at (8.1,0) {$H$};
\end{tikzpicture}
\caption{$G$: a graph with its owner blocks. $H$: its contact graph.
Here $k_x=3$.}
\label{fig:owners}
\end{figure}
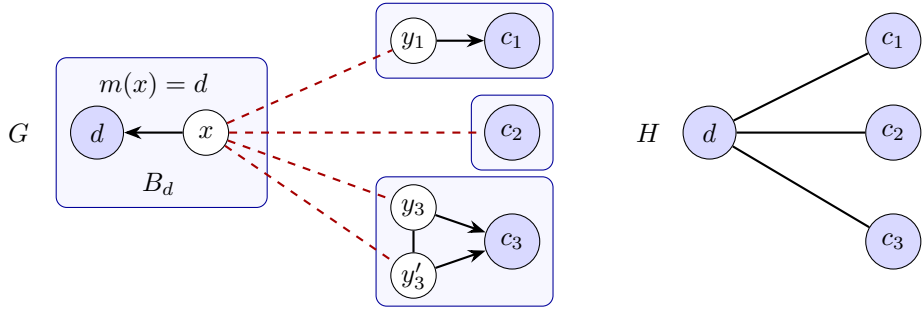

\begin{definition}
\label{def:contact}
For the fixed $D$ and $m$, the \textbf{contact graph} $H$ has vertex set $D$.
Distinct centers $d,c\in D$ are adjacent in $H$ if and only if an edge of $G$
joins their owner blocks $B_d$ and $B_c$.
\end{definition}

See graph $H$ in \cref{fig:owners}.
To obtain $H$, contract each owner star to its center, delete loops,
and keep only one edge between each pair of centers.
If $G$ is planar, then $H$ is planar, because edge contractions and
deletions preserve planarity.

For every $x\notin D$, we have $\Cout(x)\subseteq N_H(m(x))$.
Paths in $H$ lift to connected subgraphs of $G$: choose an edge between
each pair of consecutive blocks and connect its endpoints through the
owner stars. Thus the connected components of $G$ correspond to those
of $H$. The same argument after deleting a block shows that $G-B_d$
is connected whenever $H-d$ is connected.
The graph $H$ is \emph{2-connected} if $|V(H)|\ge3$ and it remains
connected after any one vertex is deleted.

\begin{theorem}
\label{thm:structure}
For every planar graph $G$, dominating set $D$, and owner assignment $m$,
\[
 W=\sum_{x\in V(G)\setminus D}(k_x-2)^+\le(4|D|-12)^+.
\]
If $W>0$, this gives $W\le4|D|-12$.
\end{theorem}

The proof has two parts. In \cref{sec:triangulations}, we prove the bound
when the contact graph $H$ is 2-connected. In \cref{sec:cuts}, we split
at cut vertices and then consider disconnected graphs to complete the proof.

\section{Tricolored faces and the 2-connected case}
\label{sec:triangulations}

The next lemma combines the standard extension to a triangulation
\cite{Diestel2025} with monotonicity of the contact counts.

\begin{lemma}
\label{lem:triangulation}
Let $G$ be a planar graph with $|V(G)| \geq 3$, and let $D$ be a
dominating set of $G$ with fixed owner blocks.
If the contact graph $H$ of $G$ is 2-connected, then there exists a
graph $G^+$ with $V(G^+)=V(G)$ and $E(G)\subseteq E(G^+)$ such that:
\begin{enumerate}[(i)]
\item $G^+$ has a plane embedding in which every face,
including the outer face, is a triangle.
\item The set $D$ remains dominating in $G^+$, with the same owner blocks.
The contact graph of $G^+$ is 2-connected.
\item Each non-center has at least as many foreign contacts in $G^+$
as in $G$, so the total weight does not decrease.
\end{enumerate}
\end{lemma}
\begin{proof}
Extend a plane embedding of $G$ to a triangulation $G^+$ on the same
vertices. All owner edges remain, so $D$ dominates $G^+$ with the same
blocks. Since $N_G(x)\subseteq N_{G^+}(x)$, no foreign contact count
decreases. The contact graph of $G^+$ contains $H$ as a spanning
subgraph, so it remains 2-connected.
\end{proof}

\begin{definition}
\label{def:tricolored-corners}
Let $G$ be a plane graph whose vertices are partitioned into blocks.
A triangular face is \textbf{tricolored} if its three vertices lie
in different blocks. A \textbf{tricolored corner at $x$} is a pair
$(f,x)$, where $f$ is a tricolored face and $x$ is one of its vertices.
We write $F_3$ for the number of tricolored faces, counting the outer
face if it is tricolored, and $t(x)$ for the number of tricolored corners at $x$.
For an owner block $B_d$, the number of tricolored corners in the block is
\[
 n_d=\sum_{x\in B_d}t(x) \quad(d\in D).
\]
This sum includes the center $d$.
\end{definition}

Since $D$ is unchanged and $W$ does not decrease, it suffices to prove
the bound for the triangulation $G^+$.
Each tricolored face has three corners, one at each vertex.
We show that the total weight in each owner block is at most the number
of tricolored corners in that block minus two.
Summing over the blocks gives $W\le3F_3-2|D|$.
We then prove $F_3\le2|D|-4$, which gives $W\le4|D|-12$.
We start by showing that the foreign neighbors of a non-center form
one interval around it. Within this interval, each change of owner
between consecutive neighbors gives a tricolored face.

\begin{definition}
\label{def:cyclic-interval}
The \textbf{cyclic order} of the neighbors of a vertex in a plane graph
is their clockwise order around that vertex. An \textbf{interval} is
a consecutive part of this cyclic order, possibly empty or the whole order.
\end{definition}

\begin{lemma}
\label{lem:arc}
Let $G$ be a plane graph with contact graph $H$, and let
$x\in B_d\setminus d$. If $H-d$ is connected, then the set
$N_G(x)\setminus B_d$ is an interval, possibly empty, in the cyclic
order of $N_G(x)$ around $x$.
\end{lemma}
\begin{proof}
If the foreign neighbors of $x$ do not form an interval, then the cyclic
order contains $d,u,p,w$, where $p\in B_d$ and $u,w\notin B_d$.
The owner edges $xd,dp$ and the edge $px$ form a triangle in $B_d$.
The order at $x$ puts $u$ and $w$ on opposite sides of this triangle
(\cref{fig:interval-triangle}). Since $H-d$ is connected, so is $G-B_d$.
A path from $u$ to $w$ in $G-B_d$ would cross the triangle, contradicting
planarity. Thus the foreign neighbors of $x$ form an interval.
\end{proof}

\begin{figure}[!htbp]
\centering
\begin{tikzpicture}[every node/.style={font=\small},
  vertex/.style={circle,draw,fill=white,minimum size=6mm,inner sep=0pt},
  center/.style={vertex,fill=blue!15}]
\coordinate (x) at (0,0);
\coordinate (d) at (3,.7);
\coordinate (p) at (3,-.7);
\coordinate (u) at (1.7,0);
\coordinate (w) at (-1.2,0);
\draw[very thick] (x)--(d)--(p)--cycle;
\draw (x)--(u);
\draw (x)--(w);
\draw[dashed,gray] (x) circle[radius=.4];
\draw[thick,-{Stealth},gray] (65:.4)
  arc[start angle=65,end angle=35,radius=.4];
\node[vertex,fill=red!10] at (x) {$x$};
\node[center] at (d) {$d$};
\node[vertex] at (p) {$p$};
\node[vertex] at (u) {$u$};
\node[vertex] at (w) {$w$};
\end{tikzpicture}
\caption{The triangle $C=xdpx$ separates $u$ and $w$.}
\label{fig:interval-triangle}
\end{figure}
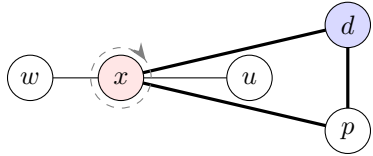

For the next lemmas, up to and including
\cref{lem:save-two}, let $G$ be a simple plane triangulation with
at least four vertices. Fix a dominating set $D$ and an owner
assignment $m$. The contact graph $H$ and all contact counts $k_x$
refer to this graph $G$.
If $|D|\le3$, then $W=0$. We therefore use these lemmas in the final
proof only when $|D|\ge4$, which also gives $|V(G)|\ge4$.

We now bound $t(x)$ from below in terms of the number $k_x$ of
foreign blocks seen by $x$.

\begin{lemma}
\label{lem:vertex-corners}
Let $x\in V(G)\setminus D$ and put $d=m(x)$.
If $H-d$ is connected, then
\[
 t(x)\ge(k_x-1)^+.
\]
\end{lemma}
\begin{proof}
If $k_x=0$, the claim follows from $t(x)\ge0$. Otherwise,
\cref{lem:arc} lets us list all neighbors outside $B_d$ in interval
order as $y_1,\ldots,y_s$. Their owner sequence contains exactly
$k_x$ distinct labels, none equal to $d$. Each label after the first
causes a change when it first appears, so there are at least $k_x-1$
changes.

At each change, the consecutive neighbors $y_i,y_{i+1}$ form a face
with $x$. Its three owners are distinct, so it is tricolored.
Different consecutive pairs give different faces in a simple plane
triangulation on at least four vertices. Hence $t(x)\ge k_x-1$.
\end{proof}

\begin{lemma}
\label{lem:boundary}
Assume that $H$ is 2-connected and let $d\in D$.
Label each edge $xu$ with $x\in B_d$ and $u\notin B_d$ by $m(u)$.
There is a cyclic ordering of these edges with exactly $n_d$ label
changes, including any change between the last and first edges.
\end{lemma}

\begin{proof}
The owner star makes $G[B_d]$ connected. Since $H$ is 2-connected,
$H-d$ is connected. Paths in $H-d$ lift through the other connected
owner blocks, so $G-B_d$ is connected. Thus all edges leaving $B_d$
enter the same face $F$ of $G[B_d]$.

We consider a thin closed neighborhood of $G[B_d]$, formed by vertex disks
and edge strips. Since $G[B_d]$ is connected, this neighborhood has
one boundary circle $C$ in $F$. The disks and strips are small enough
that $C$ crosses each leaving edge exactly once. Walking around $C$
gives the required cyclic order.

Each edge $vw$ leaving $B_d$, with $v\in B_d$ and $w\notin B_d$,
has label $m(w)$.
After crossing $vw$, the curve $C$ passes through the incident triangular face
$vwz$. If $z\notin B_d$, the next crossing is with $vz$, and the label
changes exactly when $m(w)\ne m(z)$. Since $m(v)=d$, this is exactly a
tricolored corner at $v$. If $z\in B_d$, the next crossing is with $zw$,
so the label stays $m(w)$.
Conversely, a tricolored face $vwz$ with $v\in B_d$ has
$w,z\notin B_d$ and $m(w)\ne m(z)$. Walking around $C$ passes between
$vw$ and $vz$ exactly once, so its corner at $v$ gives one label change.
Thus the total number of label changes is $n_d$
(See \cref{fig:boundary-order}).
\end{proof}

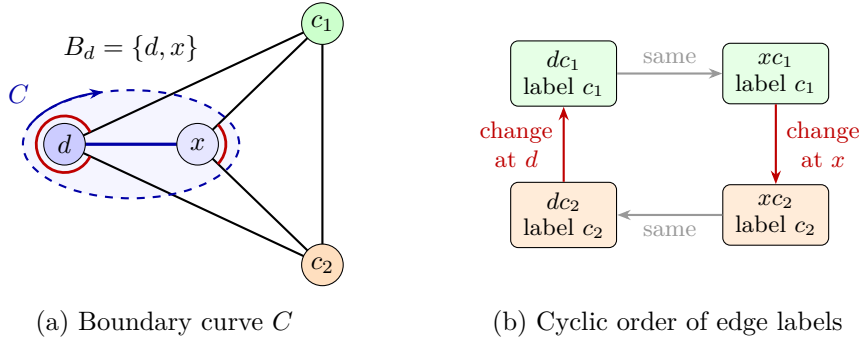
\begin{figure}[!htbp]
\centering
\begin{tikzpicture}[scale=1.1,every node/.style={font=\small},
  vertex/.style={circle,draw,minimum size=5.5mm,inner sep=0pt},
  labelbox/.style={rounded corners=1mm,draw,align=center,
    minimum width=14mm,minimum height=8mm,font=\footnotesize},
  order/.style={-{Stealth[length=1.8mm]},thick},
  change/.style={order,red!75!black},
  boundary/.style={blue!65!black,dashed,thick}]
\coordinate (bd) at (0,0);
\coordinate (bx) at (1.6,0);
\coordinate (bc1) at (3.1,1.45);
\coordinate (bc2) at (3.1,-1.45);
\fill[blue!4] (.8,0) ellipse[x radius=1.3,y radius=.65];
\draw[thick] (bd)--(bc1)--(bc2)--(bd);
\draw[thick] (bc1)--(bx)--(bc2);
\draw[blue!65!black,very thick] (bd)--(bx);
\draw[boundary] (.8,0) ellipse[x radius=1.3,y radius=.65];
\draw[blue!65!black,thick,-{Stealth[length=1.8mm]}]
  ({.8+1.3*cos(160)},{.65*sin(160)})
  arc[start angle=160,end angle=105,x radius=1.3,y radius=.65];
\node[text=blue!65!black] at (-.55,.6) {$C$};
\node at (.8,1.15) {$B_d=\{d,x\}$};
\draw[red!75!black,line width=1.05pt]
  ({.34*cos(25.1)},{.34*sin(25.1)})
  arc[start angle=25.1,end angle=334.9,radius=.34];
\draw[red!75!black,line width=1.05pt]
  ({1.6+.34*cos(-44)},{.34*sin(-44)})
  arc[start angle=-44,end angle=44,radius=.34];
\node[vertex,fill=blue!20] at (bd) {$d$};
\node[vertex,fill=blue!10] at (bx) {$x$};
\node[vertex,fill=green!20] at (bc1) {$c_1$};
\node[vertex,fill=orange!25] at (bc2) {$c_2$};
\node at (1.2,-2.15) {(a) Boundary curve $C$};
\node[labelbox,fill=green!10] (e1) at (6,.85)
  {$dc_1$\\[-1pt]label $c_1$};
\node[labelbox,fill=green!10] (e2) at (8.55,.85)
  {$xc_1$\\[-1pt]label $c_1$};
\node[labelbox,fill=orange!15] (e3) at (8.55,-.85)
  {$xc_2$\\[-1pt]label $c_2$};
\node[labelbox,fill=orange!15] (e4) at (6,-.85)
  {$dc_2$\\[-1pt]label $c_2$};
\draw[order,gray!80] (e1)--node[above,font=\footnotesize]{same}(e2);
\draw[change] (e2)--node[right,align=center,font=\footnotesize]
  {change\\at $x$}(e3);
\draw[order,gray!80] (e3)--node[below,font=\footnotesize]{same}(e4);
\draw[change] (e4)--node[left,align=center,font=\footnotesize]
  {change\\at $d$}(e1);
\node at (7.25,-2.15) {(b) Cyclic order of edge labels};
\end{tikzpicture}
\caption{The curve $C$ around $B_d=\{d,x\}$ and the cyclic edge labels.
The two tricolored faces have vertices $d,c_1,c_2$ (the outer face)
and $x,c_1,c_2$.
The two label changes give $n_d=t(d)+t(x)=2$.}
\label{fig:boundary-order}
\end{figure}

\begin{lemma}
\label{lem:block-three}
If $H$ is 2-connected and $|D|\ge3$, then for every $d\in D$:
\begin{enumerate}[(i)]
\item $n_d\ge2$.
\item $n_d\ge k_x$ for every $x\in B_d\setminus\{d\}$.
\item $\displaystyle\sum_{x\in B_d\setminus\{d\}}(k_x-1)^+\le n_d$.
\end{enumerate}
\end{lemma}

\begin{proof}
In the cyclic order from \cref{lem:boundary}, every distinct label has
a last occurrence before a different label. Thus the number $n_d$ of
changes is at least the number of distinct labels. There are at least
two labels because $d$ has degree at least two in $H$. This proves (i).
All $k_x$ labels seen by any $x\in B_d\setminus\{d\}$ occur in the same
order, giving (ii). Finally, \cref{lem:vertex-corners} gives
\[
 \sum_{x\in B_d\setminus\{d\}}(k_x-1)^+
 \le\sum_{x\in B_d\setminus\{d\}}t(x)
 =n_d-t(d)\le n_d,
\]
which proves (iii).
\end{proof}

\begin{lemma}
\label{lem:save-two}
If $H$ is 2-connected and $|D|\ge3$, then for every $d\in D$,
\[
 \sum_{x\in B_d\setminus\{d\}}(k_x-2)^+\le n_d-2.
\]
\end{lemma}
\begin{proof}
Let $U_d=\{x\in B_d\setminus\{d\}:k_x\ge3\}$.
If $U_d$ is empty, the bound follows from $n_d\ge2$.
If $U_d=\{x\}$, it follows from $k_x\le n_d$.
If $|U_d|\ge2$, then
\[
 \sum_{x\in U_d}(k_x-2)
 =\sum_{x\in U_d}(k_x-1)^+-|U_d|
 \overset{(1)}{\le} n_d-|U_d|\overset{(2)}{\le} n_d-2.
\]
Inequality (1) follows from \cref{lem:block-three}(iii), since
$U_d\subseteq B_d\setminus\{d\}$ and all summands $(k_x-1)^+$ are nonnegative.
Inequality (2) uses $|U_d|\ge2$.
Vertices outside $U_d$ have weight zero, completing the proof.
\end{proof}

We now bound the total number of tricolored corners available across
all blocks. The next fact uses the standard edge estimate for simple
bipartite planar graphs.

\begin{fact}
\label{fact:faces}
Let $G$ be a connected plane graph whose vertices are partitioned
into $q$ connected blocks, where $q\ge2$. Then $G$ has at most $2q-4$
tricolored triangular faces.
\end{fact}
\begin{proof}
Let $F_3$ count these faces.
If $F_3=0$, the claim follows from $q\ge2$. Otherwise, place a new
vertex in each tricolored triangular face and join it to the three
vertices of that face. Contract a spanning tree of each block and
delete all remaining edges of $G$.
The resulting graph $J$ is planar and bipartite, with parts consisting
of the $q$ block vertices and the $F_3$ new vertices.
\Cref{fig:face-incidence} shows an example of these operations.
The graph $J$ is simple
because each new vertex has one edge to each of three distinct blocks.
It has $q+F_3\ge3$ vertices and $3F_3$ edges, so the bipartite planar
edge bound gives $3F_3\le2(q+F_3)-4$.
Rearranging yields $F_3\le2q-4$.
\end{proof}

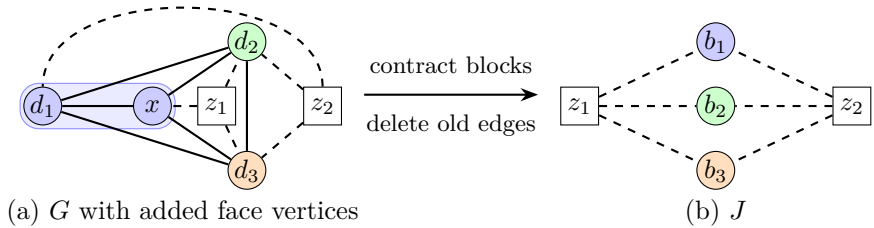
\begin{figure}[!htbp]
\centering
\begin{tikzpicture}[every node/.style={font=\small},
  vertex/.style={circle,draw,minimum size=5mm,inner sep=0pt},
  face/.style={rectangle,draw,fill=white,minimum size=5mm,inner sep=0pt},
  added/.style={thick,dashed}]
\draw[rounded corners=2.5mm,draw=blue!45,fill=blue!8]
  (-.3,-.3) rectangle (1.75,.3);
\node[vertex,fill=blue!20] (d1) at (0,0) {$d_1$};
\node[vertex,fill=blue!20] (x) at (1.45,0) {$x$};
\node[vertex,fill=green!20] (d2) at (2.7,.85) {$d_2$};
\node[vertex,fill=orange!25] (d3) at (2.7,-.85) {$d_3$};
\draw[thick] (d1)--(d2)--(d3)--(d1)--(x);
\draw[thick] (d2)--(x)--(d3);
\node[face] (leftz1) at (2.3,0) {$z_1$};
\node[face] (leftz2) at (3.7,0) {$z_2$};
\foreach \v in {x,d2,d3}{\draw[added] (leftz1)--(\v);}
\foreach \v in {d2,d3}{\draw[added] (leftz2)--(\v);}
\draw[added] (leftz2.north) .. controls (3.7,1.65) and (0,1.65)
  .. (d1.north);
\node at (1.85,-1.4) {(a) $G$ with added face vertices};
\draw[thick,-{Stealth}] (4.25,.15)--(6.55,.15);
\node[font=\footnotesize] at (5.4,.55) {contract blocks};
\node[font=\footnotesize] at (5.4,-.25) {delete old edges};
\node[face] (z1) at (7.1,0) {$z_1$};
\node[face] (z2) at (10.7,0) {$z_2$};
\node[vertex,fill=blue!20] (b1) at (8.9,.85) {$b_1$};
\node[vertex,fill=green!20] (b2) at (8.9,0) {$b_2$};
\node[vertex,fill=orange!25] (b3) at (8.9,-.85) {$b_3$};
\foreach \i in {1,2,3}{
  \draw[added] (z1)--(b\i);
  \draw[added] (z2)--(b\i);
}
\node at (8.9,-1.4) {(b) $J$};
\end{tikzpicture}
\caption{Construction of $J$ in \cref{fact:faces}, with $q=3$ and
$F_3=2$. (a) The blocks are $\{d_1,x\}$, $\{d_2\}$, and $\{d_3\}$.
Squares mark the two tricolored faces, including the outer face.
(b) Contracting the blocks and keeping only dashed edges gives $J$.}
\label{fig:face-incidence}
\end{figure}

\begin{lemma}
\label{lem:corner-budget}
Let $G$ be a simple plane triangulation with at least four vertices,
a dominating set $D$, and a fixed owner assignment.
If $|D|\ge2$, then
\[
 \sum_{d\in D}n_d=3F_3\le6|D|-12.
\]
\end{lemma}
\begin{proof}
Each tricolored face contributes three corners.
Each corner is counted in exactly one $n_d$, because its vertex belongs
to exactly one owner block. Hence $\sum_{d\in D}n_d=3F_3$. The connected owner blocks
satisfy the assumptions of \cref{fact:faces}, which gives
$3F_3\le6|D|-12$.
\end{proof}

To bound $W$ when $H$ is 2-connected, we combine the saving of two
corners per block in \cref{lem:save-two} with the total corner bound
in \cref{lem:corner-budget}.

\begin{proposition}
\label{prop:two-connected}
Let $G$ be a planar graph with dominating set $D$, fixed owner blocks,
and contact graph $H$. If $H$ is 2-connected, then $W\le4|D|-12$.
\end{proposition}
\begin{proof}
Since $H$ is 2-connected, $|D|\ge3$. If $|D|=3$, every non-center
has at most two foreign contacts, so $W=0=4|D|-12$.

Suppose that $|D|\ge4$. By \cref{lem:triangulation}, add edges to obtain
a simple plane triangulation $G^+$ with the same owner blocks and a
2-connected contact graph. Let $k_x^+$, $n_d^+$, and $W^+$ denote its
contact counts, corner counts, and total weight. Adding edges does not
decrease the contact counts, so $W\le W^+$. Now
\cref{lem:save-two,lem:corner-budget} give
\[
 W\le W^+
 =\sum_{d\in D}\sum_{x\in B_d\setminus\{d\}}(k_x^+-2)^+
 \le\sum_{d\in D}(n_d^+-2)
 \le(6|D|-12)-2|D|=4|D|-12.
\]
\end{proof}

\section{Cut vertices and disconnected contact graphs}
\label{sec:cuts}

The local corner inequality (\cref{lem:vertex-corners}) needs connectivity
of $H-d$. To extend the bound from \cref{prop:two-connected} to every
contact graph, we first split a connected graph at a cut vertex.
This gives smaller owner instances, and only the weights of vertices
owned by the cut vertex can change. We bound the total loss in these
weights using planarity, then apply induction on the number of centers.
Finally, we sum over the connected components of $H$ to prove
\cref{thm:structure} and derive the six-contact bound.

\subsection{Splitting at a cut vertex}

We first describe the smaller instances and identify which contact
counts stay the same. Throughout this subsection, $H$ is connected
and $v$ is a cut vertex of $H$.

\begin{definition}
\label{def:cut-pieces}
Let $C_1,\ldots,C_\ell$, with $\ell\ge2$, be the vertex sets of the
connected components of $H-v$. For each $i$, set
\[
 X_i=C_i\cup\{v\},\quad V_i=\bigcup_{c\in X_i}B_c,\quad
 G_i=G[V_i].
\]
Each vertex keeps its original owner: define $m_i:V_i\to X_i$ by
$m_i(x)=m(x)$ for every $x\in V_i$.
We call $(G_i,X_i,m_i)$ the $i$th \textbf{piece} of the decomposition at $v$.
Let $W_i$ denote its total weight, defined as in \cref{def:contact-weight}.
\end{definition}

To form $H[X_i]$, we take the component $H[C_i]$ and add $v$ with its
edges to that component. To form $G_i$, we keep all owner blocks with
centers in $C_i$ and the whole block $B_v$. Two different graphs $G_i$
and $G_j$ have exactly the vertices of $B_v$ in common, so
$V_i\cap V_j=B_v$ for $i\ne j$.
The construction from \cref{def:cut-pieces} is shown in \cref{fig:cut-pieces}.

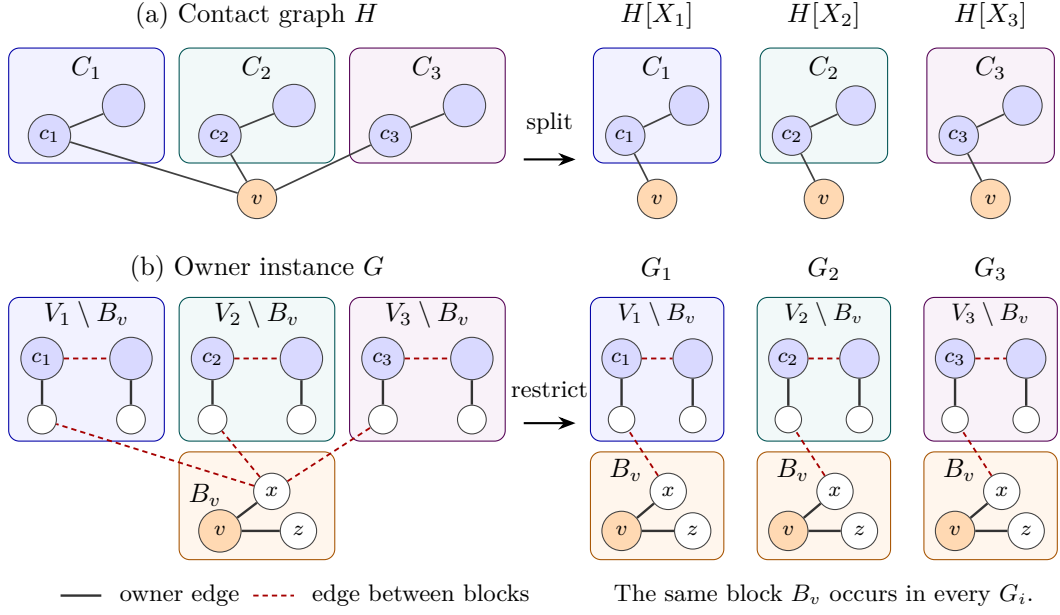
\begin{figure}[!t]
\centering
\begin{tikzpicture}[x=.98cm,y=.95cm,every node/.style={font=\small},
  ctr/.style={circle,draw=black!75,fill=blue!15,minimum size=5.2mm,
    inner sep=0pt,font=\scriptsize},
  vert/.style={circle,draw=black!75,fill=white,minimum size=4.8mm,
    inner sep=0pt,font=\scriptsize},
  leaf/.style={circle,draw=black!75,fill=white,minimum size=3.5mm,
    inner sep=0pt},
  cut/.style={ctr,fill=orange!30},
  owner/.style={draw=black!80,line width=.9pt},
  contact/.style={draw=red!65!black,dash pattern=on 2pt off 1.5pt,line width=.7pt},
  hedge/.style={draw=black!75,line width=.65pt},
  shared/.style={rounded corners=4pt,draw=orange!65!black,fill=orange!7},
  heading/.style={font=\small,align=center}]
\begin{scope}[yshift=3.25cm]
  \node[heading] at (3.3,3.3) {(a) Contact graph $H$};
  \node[cut] (hv) at (3.3,.75) {$v$};
  \foreach \i/\xx/\col in {1/1/blue,2/3.3/teal,3/5.6/violet}{
    \draw[rounded corners=4pt,draw=\col!65!black,fill=\col!5]
      (\xx-1.05,1.25) rectangle (\xx+1.05,2.85);
    \node at (\xx,2.58) {$C_\i$};
    \node[ctr,minimum size=5.6mm] (hc\i) at (\xx-.5,1.63) {$c_\i$};
    \node[ctr,minimum size=5.6mm] (hd\i) at (\xx+.5,2.05) {};
    \draw[hedge] (hc\i)--(hd\i) (hc\i)--(hv);
  }
  \draw[-{Stealth[length=2.3mm]},thick] (6.9,1.3)--(7.6,1.3)
    node[midway,above=2mm,font=\footnotesize] {split};
  \foreach \i/\xx/\col in {1/8.7/blue,2/10.95/teal,3/13.2/violet}{
    \node[heading] at (\xx,3.3) {$H[X_\i]$};
    \draw[rounded corners=4pt,draw=\col!65!black,fill=\col!5]
      (\xx-.86,1.25) rectangle (\xx+.86,2.85);
    \node at (\xx,2.58) {$C_\i$};
    \node[ctr] (pc\i) at (\xx-.43,1.63) {$c_\i$};
    \node[ctr] (pd\i) at (\xx+.43,2.05) {};
    \node[cut] (pv\i) at (\xx,.75) {$v$};
    \draw[hedge] (pc\i)--(pd\i) (pc\i)--(pv\i);
  }
\end{scope}
\node[heading] at (3.3,3.2) {(b) Owner instance $G$};
\foreach \i/\xx/\col in {1/1/blue,2/3.3/teal,3/5.6/violet}{
  \draw[rounded corners=4pt,draw=\col!65!black,fill=\col!5]
    (\xx-1.05,.8) rectangle (\xx+1.05,2.8);
  \node[inner sep=1pt] at (\xx,2.57)
    {$V_\i\setminus B_v$};
  \node[ctr,minimum size=5.6mm] (gc\i) at (\xx-.6,1.95) {$c_\i$};
  \node[ctr,minimum size=5.6mm] (gd\i) at (\xx+.6,1.95) {};
  \node[leaf,minimum size=3.8mm] (gy\i) at (\xx-.6,1.1) {};
  \node[leaf,minimum size=3.8mm] (gz\i) at (\xx+.6,1.1) {};
  \draw[owner] (gc\i)--(gy\i) (gd\i)--(gz\i);
  \draw[contact] (gc\i)--(gd\i);
}
\draw[shared] (2.25,-.85) rectangle (4.35,.65);
\node at (2.6,.05) {$B_v$};
\node[cut,minimum size=5.6mm] (gv) at (2.8,-.45) {$v$};
\node[vert] (gx) at (3.5,.1) {$x$};
\node[vert] (gz) at (3.86,-.45) {$z$};
\draw[owner] (gv)--(gx) (gv)--(gz);
\foreach \i in {1,2,3}{\draw[contact] (gx)--(gy\i);}
\draw[-{Stealth[length=2.3mm]},thick] (6.9,1.05)--(7.6,1.05)
  node[midway,above=2mm,font=\footnotesize] {restrict};
\foreach \i/\xx/\col in {1/8.7/blue,2/10.95/teal,3/13.2/violet}{
  \node[heading] at (\xx,3.2) {$G_\i$};
  \draw[rounded corners=4pt,draw=\col!65!black,fill=\col!5]
    (\xx-.9,.8) rectangle (\xx+.9,2.8);
  \node[font=\footnotesize,inner sep=1pt] at (\xx,2.57)
    {$V_\i\setminus B_v$};
  \node[ctr] (ic\i) at (\xx-.48,1.95) {$c_\i$};
  \node[ctr] (id\i) at (\xx+.48,1.95) {};
  \node[leaf] (iy\i) at (\xx-.48,1.1) {};
  \node[leaf] (iz\i) at (\xx+.48,1.1) {};
  \draw[owner] (ic\i)--(iy\i) (id\i)--(iz\i);
  \draw[contact] (ic\i)--(id\i);
  \draw[shared] (\xx-.9,-.85) rectangle (\xx+.9,.65);
  \node at (\xx-.42,.4) {$B_v$};
  \node[cut] (iv\i) at (\xx-.48,-.45) {$v$};
  \node[vert] (ix\i) at (\xx+.15,.1) {$x$};
  \node[vert] (izb\i) at (\xx+.48,-.45) {$z$};
  \draw[owner] (iv\i)--(ix\i) (iv\i)--(izb\i);
  \draw[contact] (ix\i)--(iy\i);
}
\draw[owner] (.65,-1.35)--(1.2,-1.35)
  node[right=1mm,font=\footnotesize] {owner edge};
\draw[contact] (3.25,-1.35)--(3.8,-1.35)
  node[right=1mm,font=\footnotesize] {edge between blocks};
\node[font=\footnotesize] at (10.95,-1.35)
  {The same block $B_v$ occurs in every $G_i$.};
\end{tikzpicture}
\caption{Splitting at $v$: (a) contact graphs, (b) induced owner instances.
The pieces share the same block $B_v$, while the contacts of $x$ split
between them.}
\label{fig:cut-pieces}
\end{figure}

\begin{lemma}
\label{lem:induced-pieces}
Each piece is a planar owner instance with the original blocks indexed
by $X_i$. Its contact graph $H[X_i]$ is connected and has at least two
vertices. Moreover,
\[
 \sum_i |X_i|=|D|-1+\ell.
\]
Every non-center whose owner lies in $C_i$ appears only in the $i$th
piece and keeps its foreign contact count there.
\end{lemma}
\begin{proof}
Each $G_i$ is an induced subgraph of $G$ containing whole owner blocks,
so it is planar and keeps every owner edge. Its contact graph is
$H[X_i]$. This graph is connected because $H[C_i]$ is connected and
has an edge to $v$.

The nonempty sets $C_i$ partition $D\setminus\{v\}$, and $\ell\ge2$.
Hence $2\le|X_i|<|D|$ and
\[
 \sum_i|X_i|=\sum_i(|C_i|+1)=|D|-1+\ell.
\]
If $x$ has owner $c\in C_i$, every foreign contact has its center in
$N_H(c)\subseteq C_i\cup\{v\}$. Thus all its contacts remain in $G_i$,
and its block occurs in no other piece.
\end{proof}

By \cref{lem:induced-pieces}, we only need to control the weights of
non-centers in $B_v$. This block occurs in every piece, while each of
its foreign contacts belongs to exactly one component of $H-v$.

\begin{definition}
\label{def:piece-contacts}
For $x\in B_v\setminus\{v\}$, define
\[
 k_x^{(i)}=|\Cout(x)\cap C_i|,\quad
 I_x=\{i\in\{1,\ldots,\ell\}:k_x^{(i)}>0\}.
\]
Thus $k_x^{(i)}$ is the foreign contact count of $x$ in $G_i$, and
$k_x=\sum_i k_x^{(i)}$.
We call $I_x$ the \textbf{set of active indices} of $x$.
An index $i$ is active exactly when $x$ has a neighbor in $V_i\setminus B_v$.
\end{definition}

The next lemma bounds the loss in the weight of one such vertex.

\begin{lemma}
\label{lem:split-charge}
For every $x\in B_v\setminus\{v\}$,
\[
 (k_x-2)^+\le\sum_i(k_x^{(i)}-2)^++2(|I_x|-1)^+.
\]
\end{lemma}
\begin{proof}
If $|I_x|\le1$, all contacts lie in at most one piece, and the bound
is an equality. Otherwise $|I_x|\ge2$.
For each $i\in I_x$, we have $k_x^{(i)}\ge1$, and
$k_x=\sum_{i\in I_x}k_x^{(i)}$. Hence $2\le|I_x|\le k_x$, and
\[
 \sum_i(k_x^{(i)}-2)^+
 =\sum_{i\in I_x}(k_x^{(i)}-2)^+
 \ge\sum_{i\in I_x}(k_x^{(i)}-2)=k_x-2|I_x|.
\]
Adding $2(|I_x|-1)$ gives $k_x-2=(k_x-2)^+$ on the right,
which proves the bound.
\end{proof}

We now bound the sum of the terms $2(|I_x|-1)^+$ from
\cref{lem:split-charge} over all $x\in B_v\setminus\{v\}$.
To do this, we count pairs $(x,i)$ with $i\in I_x$, that is,
with $x$ having a neighbor in $V_i\setminus B_v$.
The owner edges at $v$ allow us to represent these incidences in a
simple bipartite planar graph. We then apply the standard edge bound
from Euler's formula \cite{Diestel2025}.

\begin{lemma}
\label{lem:cross}
Let $Z=\{x\in B_v\setminus\{v\}:|I_x|\ge2\}$. Then
\mbox{$\sum_{x\in Z}(|I_x|-1)\le2(\ell-1)$}, where $\ell$ is the number
of connected components of $H-v$.
\end{lemma}
\begin{proof}
For each component $C_i$ of $H-v$, the blocks $B_c$ with $c\in C_i$
together induce a connected subgraph of $G$, since each block and
$H[C_i]$ are connected.
Contract this subgraph to a vertex $c_i$.
Keep $v$, the vertices of $Z$, and all contracted vertices
$c_1,\ldots,c_\ell$, together with the owner edges $vx$ for $x\in Z$
and one edge $xc_i$ for each $x\in Z$ and $i\in I_x$.
Delete all other vertices and edges. A contracted vertex with no
neighbor in $Z$ remains as an isolated vertex.
See \cref{fig:cross-piece-contraction} for an example.

This is a simple planar bipartite graph with parts
$Z$ and $\{v,c_1,\ldots,c_\ell\}$. It has $|Z|+\ell+1$ vertices.
There are $|Z|$ owner edges from $v$ to $Z$, and each $x\in Z$ has
$|I_x|$ edges to the contracted vertices. Thus the number of edges is
$|Z|+\sum_{x\in Z}|I_x|$. Since $\ell\ge2$, the graph has at least
three vertices. The bipartite planar edge bound applies, including
when the graph has isolated vertices, and gives
\[
 |Z|+\sum_{x\in Z}|I_x|\le2(|Z|+\ell+1)-4.
\]
Subtracting $2|Z|$ gives the claimed bound, also when $Z=\varnothing$.
\end{proof}

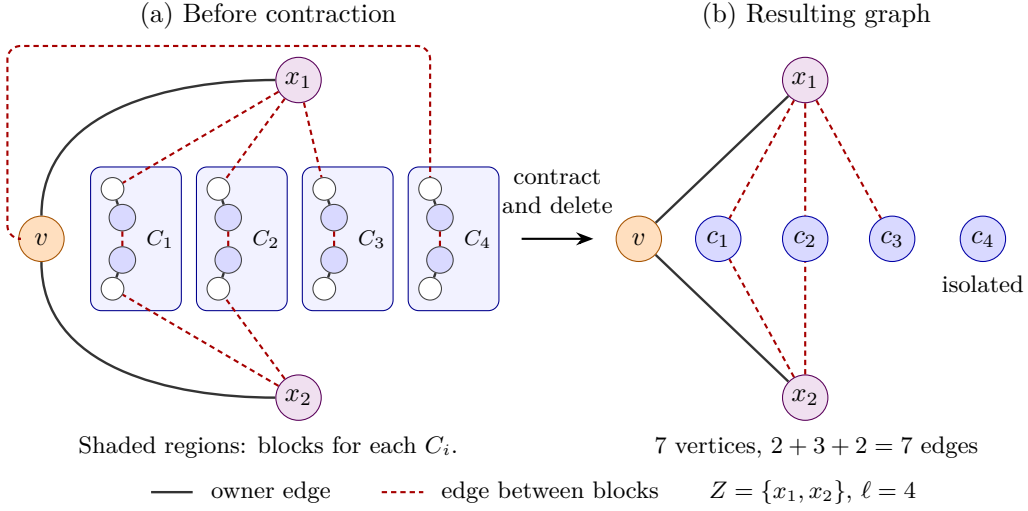
\begin{figure}[!ht]
\centering
\begin{tikzpicture}[x=1cm,y=1cm,every node/.style={font=\small},
  center/.style={circle,draw=black!75,fill=blue!15,minimum size=3.5mm,
    inner sep=0pt},
  leaf/.style={circle,draw=black!75,fill=white,minimum size=3mm,
    inner sep=0pt},
  cross/.style={circle,draw=violet!70!black,fill=violet!12,
    minimum size=6mm,inner sep=0pt},
  contracted/.style={circle,draw=blue!65!black,fill=blue!15,
    minimum size=6mm,inner sep=0pt},
  cut/.style={circle,draw=orange!65!black,fill=orange!25,
    minimum size=6mm,inner sep=0pt},
  owner/.style={draw=black!80,line width=.9pt},
  contact/.style={draw=red!65!black,dash pattern=on 2pt off 1.5pt,
    line width=.8pt}]
\node at (2.1,2.95) {(a) Before contraction};
\node[cross] (lx1) at (2.5,2.1) {$x_1$};
\node[cross] (lx2) at (2.5,-2.1) {$x_2$};
\node[cut] (lv) at (-.9,0) {$v$};
\draw[owner] (lv) to[out=90,in=180] (lx1);
\draw[owner] (lv) to[out=-90,in=180] (lx2);
\foreach \i/\xx in {1/.35,2/1.75,3/3.15,4/4.55}{
  \draw[rounded corners=4pt,draw=blue!55!black,fill=blue!5]
    (\xx-.6,-.95) rectangle (\xx+.6,.95);
  \node[font=\scriptsize] at (\xx+.32,0) {$C_\i$};
  \node[center] (la\i) at (\xx-.18,.28) {};
  \node[center] (lb\i) at (\xx-.18,-.28) {};
  \node[leaf] (lt\i) at (\xx-.31,.66) {};
  \node[leaf] (ld\i) at (\xx-.31,-.66) {};
  \draw[owner] (la\i)--(lt\i) (lb\i)--(ld\i);
  \draw[contact] (la\i)--(lb\i);
}
\foreach \i in {1,2,3}{\draw[contact] (lx1)--(lt\i);}
\foreach \i in {1,2}{\draw[contact] (lx2)--(ld\i);}
\draw[contact,rounded corners=5pt] (lv)--(-1.35,0)--(-1.35,2.55)
  --(4.24,2.55)--(lt4);
\node[font=\footnotesize] at (2.1,-2.75)
  {Shaded regions: blocks for each $C_i$.};
\draw[-{Stealth[length=2.3mm]},thick] (5.45,0)--(6.4,0)
  node[midway,above=2mm,align=center,font=\footnotesize] {contract\\and delete};
\begin{scope}[xshift=7.1cm]
  \node at (2.25,2.95) {(b) Resulting graph};
  \node[cross] (rx1) at (2.1,2.1) {$x_1$};
  \node[cross] (rx2) at (2.1,-2.1) {$x_2$};
  \node[cut] (rv) at (-.1,0) {$v$};
  \draw[owner] (rv)--(rx1) (rv)--(rx2);
  \foreach \i/\xx in {1/.95,2/2.1,3/3.25,4/4.45}{
    \node[contracted] (rc\i) at (\xx,0) {$c_\i$};
  }
  \foreach \i in {1,2,3}{\draw[contact] (rx1)--(rc\i);}
  \foreach \i in {1,2}{\draw[contact] (rx2)--(rc\i);}
  \node[font=\footnotesize] at (4.45,-.6) {isolated};
  \node[font=\footnotesize] at (2.25,-2.75)
    {$7$ vertices, $2+3+2=7$ edges};
\end{scope}
\draw[owner] (.55,-3.35)--(1.1,-3.35)
  node[right=1mm,font=\footnotesize] {owner edge};
\draw[contact] (3.6,-3.35)--(4.15,-3.35)
  node[right=1mm,font=\footnotesize] {edge between blocks};
\node[font=\footnotesize] at (9.3,-3.35) {$Z=\{x_1,x_2\}$, $\ell=4$};
\end{tikzpicture}
\caption{Construction in \cref{lem:cross}.
(a) Owner instance with one shaded region for each $C_i$.
(b) Result after contracting the regions and keeping the owner edges
from $v$ to $Z$ and the contacts between $Z$ and the contracted vertices.}
\label{fig:cross-piece-contraction}
\end{figure}

Combining the unchanged contact counts outside $B_v$ with
\cref{lem:split-charge,lem:cross} gives the estimate needed for induction.

\begin{corollary}
\label{cor:global-split}
For the pieces from \cref{def:cut-pieces},
\[
 W\le\sum_iW_i+4(\ell-1).
\]
\end{corollary}
\begin{proof}
By \cref{lem:induced-pieces}, every non-center outside $B_v$ occurs in
exactly one piece and keeps its weight. The total weight of these vertices
is the same in $W$ and in $\sum_iW_i$.
For a non-center in $B_v$ with at most one active index, all contacts
lie in at most one piece, so its weight is also unchanged.
Thus only vertices in $Z$ contribute to $W-\sum_iW_i$.
By \cref{lem:split-charge}, we obtain
\[
 W-\sum_iW_i=\sum_{x\in Z}
 \left((k_x-2)^+-\sum_i(k_x^{(i)}-2)^+\right)
 \le2\sum_{x\in Z}(|I_x|-1).
\]
By \cref{lem:cross}, $\sum_{x\in Z}(|I_x|-1)\le2(\ell-1)$.
Adding $\sum_iW_i$ to both sides gives the claim.
\end{proof}

\subsection{Connected contact graphs}

We now combine \cref{cor:global-split} with the 2-connected case.
The size identity in \cref{lem:induced-pieces} accounts for the center
$v$ shared by all pieces.

\begin{proposition}
\label{prop:connected}
If $H$ is connected, then $W\le(4|D|-12)^+$.
\end{proposition}
\begin{proof}
We use strong induction on $|D|$. Each non-center can have contacts
with at most $|D|-1$ foreign blocks, since its own block is excluded.
Thus, if $|D|\le3$, every weight $(k_x-2)^+$ is zero, so $W=0$.
Now let $|D|\ge4$ and assume that the statement holds for every owner
instance with a connected contact graph and fewer centers.
If $H$ has no cut vertex, it is 2-connected and
\cref{prop:two-connected} applies.

Otherwise, $H$ is connected but not 2-connected.
Choose a cut vertex $v$ and form the pieces from
\cref{def:cut-pieces}. By \cref{lem:induced-pieces}, their contact
graphs are connected and $|X_i|<|D|$, so induction gives
$W_i\le(4|X_i|-12)^+$ for each piece.
Let $\ell_2$ be the number of pieces with $|X_i|=2$.
Every piece has at least two centers, so removing the positive parts
adds $4$ exactly for the two-center pieces:
\[
 \sum_iW_i\le\sum_i(4|X_i|-12)^+
 =\sum_i(4|X_i|-12)+4\ell_2.
\]
If at least one piece has three or more centers, then
$\ell_2\le\ell-1$. Using \cref{cor:global-split} and
$\sum_i|X_i|=|D|-1+\ell$ from \cref{lem:induced-pieces}, we obtain
\[
 W\le\sum_i(4|X_i|-12)+4\ell_2+4(\ell-1)
 =4|D|-8-4\ell+4\ell_2\le4|D|-12.
\]

We now consider the case $\ell_2=\ell$. Each component of $H-v$ is then a
single vertex, so $H$ is a star centered at $v$, with $|D|=\ell+1$.
Non-centers outside $B_v$ have at most one foreign contact and weight
zero. For $x\in B_v\setminus\{v\}$, each active piece contains just
one foreign block, so $k_x=|I_x|$. Thus
\[
 W=\sum_{x\in B_v\setminus\{v\}}(|I_x|-2)^+
 \le\sum_{x\in Z}(|I_x|-1)
 \le2\ell-2\le4(\ell+1)-12=4|D|-12.
\]
For $x\notin Z$, the term $(|I_x|-2)^+$ is zero. For $x\in Z$,
we have $|I_x|\ge2$, so $(|I_x|-2)^+=|I_x|-2\le|I_x|-1$.
This gives the first inequality. The second follows from \cref{lem:cross}.
The third holds because $\ell=|D|-1\ge3$.

\end{proof}

\subsection{The general structural theorem}

We finish the proof by applying \cref{prop:connected} to the connected
components of $H$. We then use the weight of each vertex with at least
six foreign contacts to obtain the bound needed for the algorithm.

\begin{restatedstructure}
For every planar graph $G$, dominating set $D$, and owner assignment $m$,
\[
 W=\sum_{x\in V(G)\setminus D}(k_x-2)^+\le(4|D|-12)^+.
\]
If $W>0$, this gives $W\le4|D|-12$.
\end{restatedstructure}
\begin{proof}
For each connected component of $H$, take the subgraph of $G$ formed by
all blocks whose centers lie in that component. There are no edges of $G$
between these subgraphs, so every vertex keeps the same contact count
and weight. Components with weight zero do not contribute to $W$.
Let $X_1,\ldots,X_r$ be the sets of centers in the components with
positive weight. Each such component has at least four centers, since positive weight
requires three foreign blocks as well as the vertex's own block.
By \cref{prop:connected}, its weight is at most $4|X_i|-12$.
If $W>0$, then $r\ge1$, so
\[
 W\le\sum_{i=1}^r(4|X_i|-12)\le4|D|-12r\le4|D|-12.
\]
If $W=0$, the bound $W\le(4|D|-12)^+$ holds directly.
\end{proof}

\begin{corollary}
\label{cor:threshold}
Let \mbox{$A=\{x\notin D:k_x\ge6\}$} be the set of non-centers with
contacts in at least six foreign blocks. Then
\[
 |A|\le(|D|-3)^+.
\]
\end{corollary}
\begin{proof}
Each vertex of $A$ has weight at least $4$, and all other weights
are nonnegative. Hence $4|A|\le W\le(4|D|-12)^+$ by
\cref{thm:structure}. Divide by four.
\end{proof}

The bound for $A$ is the one used in \cref{sec:algorithm}.
There, we choose $D$ to be a minimum dominating set and use
$|A|\le(|D|-3)^+$. This choice is made only in the proof.
The algorithm does not need to know $D$.

\section{The constant-round domination algorithm}
\label{sec:algorithm}

Below we describe the three-phase algorithm of Heydt et al.~\cite{HeydtEtAl2025}
with our parameters. Their constant-round implementation of the first
two phases also applies to these parameters. We work in the deterministic
$\LOCAL$ model with unique vertex identifiers.

The algorithm requires neither a plane embedding nor a dominating set
as input. For each fixed $\eps>0$, its number of rounds is independent
of the size of the graph.

\subsection{The three phases}

In the first phase, we select into $D_1$ every vertex whose open
neighborhood cannot be dominated by at most six other vertices.
All vertices in $N[D_1]$ are then dominated. The set $R_1$ contains
those still undominated, and $N_1(v)$ contains the neighbors of $v$ in $R_1$.

In the second phase, two distinct vertices are \textbf{partners} if they
have at least nineteen common neighbors in $R_1$. We select into $D_2$
every vertex that has a partner, even if it is already dominated.
The set $R$ contains the vertices still undominated by $S=D_1\cup D_2$.

In the third phase, we use the same LP-based routine as
Heydt et al.~\cite{HeydtEtAl2025}, with the parameters given in
\cref{sec:residual}, to find a set $D_3$ dominating $R$.
The routine first reduces
the remaining instance and selects vertices of high degree. It then
solves a covering linear program approximately and turns the fractional
solution into selected vertices. The output is $S\cup D_3$.

\begin{enumerate}[{Phase} 1.]
\item \(\begin{aligned}[t]
 D_1&=\{v\in V:\nexists Z\subseteq V\setminus\{v\}\ (|Z|\le6\land N(v)\subseteq N[Z])\},\\
 R_1&=V\setminus N[D_1],\qquad N_1(v)=N(v)\cap R_1.
\end{aligned}\)
\item \(\begin{aligned}[t]
 P_v&=\{z\in V\setminus\{v\}:|N_1(v)\cap N_1(z)|\ge19\},\quad
 D_2=\{v\in V:P_v\ne\varnothing\},\\
 S&=D_1\cup D_2,\qquad R=V\setminus N[S].
\end{aligned}\)
\item Use the LP-based routine to find a set $D_3$ dominating $R$.
Output $S\cup D_3$.
\end{enumerate}

The following lemma is based on Heydt et al.~\cite[Lemmas 8.1 and 8.2]{HeydtEtAl2025},
with only the cover size and partner threshold changed to six and nineteen.
The full proof is given in Appendix~\ref{app:partners-proof}.

\begin{lemma}
\label{lem:partners}
If $Z\subseteq V\setminus\{v\}$, $|Z|\le6$, and $N(v)\subseteq N[Z]$,
then $P_v\subseteq Z$.
Moreover $D_1\cap D_2=\varnothing$, and every vertex of $D_2$ has at most
six partners.
\end{lemma}

\subsection{Used in the analysis}

Fix a minimum dominating set $D$ and choose any owner assignment to $D$. These objects are used only in the proof.

\begin{definition}
\label{def:hard-vertices}
A vertex $v\in V$ is \textbf{hard with respect to $D$} if its open
neighborhood cannot be dominated by at most six vertices from
$D\setminus\{v\}$. We write $\hard$ for the set of these vertices:
\[
 \hard\coloneqq\{v\in V:\nexists Z\subseteq D\setminus\{v\}\
 (|Z|\le6\land N(v)\subseteq N[Z])\}.
\]
\end{definition}

This is analogous to the set $\widehat D$ used
in \cite{HeydtSiebertzVigny2022,HeydtEtAl2025}, with threshold
six instead of three for planar graphs.
In the latter work, the general threshold is $2\nabla-1$, and
$\nabla=2$ for planar graphs.

\begin{fact}
\label{fact:hard-partners}
We have $D_1\subseteq\hard$. For every $v\notin\hard$, we have
$P_v\subseteq D$.
\end{fact}
\begin{proof}
The first inclusion follows from the definitions of $D_1$ and $\hard$.
If $v\notin\hard$, there is a set $Z\subseteq D\setminus\{v\}$ with
$|Z|\le6$ and $N(v)\subseteq N[Z]$.
By \cref{lem:partners}, $P_v\subseteq Z\subseteq D$.
\end{proof}

We now bound $|\hard\setminus D|$ using \cref{cor:threshold}.

\begin{lemma}
\label{lem:hard}
$\hard\setminus D\subseteq A$. Thus
$|\hard\setminus D|\le(|D|-3)^+$.
If this set is nonempty, then $|D|\ge7$.
\end{lemma}
\begin{proof}
For $x\notin D$ let $M_x=\{m(y):y\in N(x)\}$.
Thus $M_x$ is the set of centers assigned to the neighbors of $x$,
or equivalently, the centers of stars containing at least one neighbor of $x$.
Every neighbor $y$ belongs to $N[m(y)]$, including when $y$ is a center,
so $N(x)\subseteq N[M_x]$.
For $x\in\hard\setminus D$, we therefore have $|M_x|\ge7$.
The owner $m(x)$ belongs to $M_x$ because $m(x)\in N(x)$ and
$m(m(x))=m(x)$. The other centers in $M_x$ correspond exactly to the foreign blocks
containing a neighbor of $x$.
Hence $k_x=|M_x|-1\ge6$ and $|D|\ge7$.
Thus $\hard\setminus D\subseteq A$. By \cref{cor:threshold}, we obtain
$|\hard\setminus D|\le|A|\le(|D|-3)^+$.
\end{proof}

We partition $D_2$ into two sets. The set $E_1$ contains the vertices
of $D_2\cap D$ and all their partners. The set $E_2$ contains the
remaining vertices of $D_2$. Thus
\begin{align*}
 E_1&=\bigcup_{v\in D_2\cap D}(\{v\}\cup P_v),\\
 E_2&=D_2\setminus E_1.
\end{align*}
These sets are used only in the analysis.

\begin{lemma}
\label{lem:partner-partition}
We have $D_2=E_1\cup E_2$, $E_1\cap E_2=\varnothing$, and
$E_2\subseteq\hard\setminus D$.
\end{lemma}
\begin{proof}
Partnership is symmetric, so every partner of a vertex in $D_2$ is also
in $D_2$. Hence $E_1\subseteq D_2$, and $E_1,E_2$ partition $D_2$.
Since $D_2\cap D\subseteq E_1$, the set $E_2$ is disjoint from $D$.
If $v\in E_2\setminus\hard$, \cref{fact:hard-partners} gives a
partner $z\in D$. Symmetry gives $z\in D_2$ and $v\in P_z$, so
$v\in E_1$, a contradiction. Thus $E_2\subseteq\hard\setminus D$.
\end{proof}

The next lemma bounds the number of vertices selected in the first two phases.

\begin{lemma}
\label{lem:phase-count}
We have
\[
 |S|\le (|D|-3)^+ + 7|S\cap D|.
\]
\end{lemma}
\begin{proof}
By \cref{fact:hard-partners,lem:partner-partition},
$S\setminus D\subseteq(\hard\setminus D)\cup(E_1\setminus D)$.
The first set has size at most $(|D|-3)^+$ by \cref{lem:hard}.
Every vertex of $E_1\setminus D$ is a partner of a vertex in $D_2\cap D$,
and each such vertex has at most six partners by \cref{lem:partners}.
Therefore
\[
 |S|\le (|D|-3)^+ +6|D_2\cap D|+|S\cap D|
 \le (|D|-3)^+ +7|S\cap D|.\qedhere
\]
\end{proof}

\section{The third phase and the approximation guarantee}
\label{sec:residual}

We bound the residual degree and optimum to apply the third phase of
Heydt et al.~\cite{HeydtEtAl2025}, then bound the full output.

\subsection{Bounds used in the third phase}

The following lemma is based on Heydt et al.~\cite[Lemma 8.7]{HeydtEtAl2025},
with cover size six and partner threshold nineteen giving the bound $6\cdot19=114$.
The full proof is given in Appendix~\ref{app:residual-degree-proof}.

\begin{lemma}
\label{lem:residual-degree}
The set $R=V\setminus N[S]$ consists of the vertices still undominated
after the first two phases, where $S=D_1\cup D_2$.
For every vertex $v\in V(G)$,
\[
 |N(v)\cap R|\le114.
\]
\end{lemma}
Let $D_R\subseteq V(G)$ be a smallest set that dominates $R$.
Its vertices may lie anywhere in $V(G)$, including outside $R$.
Like $D$, this set is used only in the proof.
The following observation also appears in \cite{HeydtEtAl2025}.

\begin{lemma}
\label{lem:residual-optimum}
The set $D\setminus S$ dominates $R$, and
\[
 |D_R|\le|D\setminus S|=|D|-|S\cap D|.
\]
\end{lemma}
\begin{proof}
Every vertex in $R=V\setminus N[S]$ is dominated by $D$ and has no
dominator in $S$. Thus $D\setminus S$ dominates $R$, and minimality
of $D_R$ gives $|D_R|\le|D\setminus S|=|D|-|S\cap D|$.
\end{proof}

\subsection{The third phase}

\begin{theorem}
\label{thm:residual}
Let $R=V(G)\setminus N[S]$ be the set left undominated after the first
two phases of the algorithm on a planar graph $G$.
For every fixed $\eps>0$, there is a deterministic $\LOCAL$ algorithm
that computes a set $D_3$ dominating $R$ with
\[
 |D_3|\le(7+\eps)|D_R|
 \le(7+\eps)(|D|-|S\cap D|)
\]
in a number of rounds depending only on $\eps$.
\end{theorem}
\begin{proof}
We apply the third phase of
Heydt et al.~\cite[Section 6, Lemma 6.3]{HeydtEtAl2025}.
By \cref{lem:residual-degree}, we can use 114 as the residual-degree
bound in their algorithm.
For planar graphs, the density bound is three, so their result gives
\[
 |D_3|\le7(1+\xi)|D_R|
\]
for every fixed $\xi>0$.
Their algorithm runs in a number of $\LOCAL$ rounds depending only
on $\xi$ and the residual-degree bound.
With $\xi=\eps/7$, we obtain $7(1+\xi)=7+\eps$.
Finally, \cref{lem:residual-optimum} gives
$|D_R|\le|D|-|S\cap D|$, which proves the second inequality.
\end{proof}

\subsection{The approximation guarantee}

We add the bounds from \cref{lem:phase-count,thm:residual} to bound the full output.

\begin{theorem}
\label{thm:algorithm-proof}
The three-phase algorithm returns a dominating set of size at most
$(8+\eps)|D|$.
\end{theorem}
\begin{proof}
By \cref{lem:phase-count,thm:residual}, we have
$|S\cup D_3|\le (|D|-3)^+ +(7+\eps)|D|-\eps|S\cap D|
\le (8+\eps)|D|$.
The output dominates $V(G)\setminus R$ through $S$ and dominates $R$
through $D_3$. The empty graph has empty output.
\end{proof}

Together with the constant-round implementation described above,
\cref{thm:algorithm-proof} proves \cref{thm:main}.

\section{Scope and further questions}
\label{sec:discussion}

The coefficient one in the bound on $|A|$ is sharp, even for a minimum
dominating set (\cref{prop:sharp}, Appendix~\ref{sec:sharpness}).
This does not prove that factor $8+\eps$ is optimal.
Closing the gap to $7$ may require a stronger joint analysis of the
first two phases and the residual instance, or a different selection rule.

\subsection{A consequence for bounded Euler genus}

The transfer theorem of Bonamy et al.~\cite{BonamyEtAl2026} improves
their ratio $34+\eps$ to $25+\eps$ using our planar algorithm.
With accuracy $\eps/3$, their Theorem 4.2 and Corollary 4.3 give
$3(8+\eps/3)+1=25+\eps$.
Our algorithm uses identifiers only for comparisons, and the required
uniformity follows from their results for dominating set.
For Euler genus $1\le g\le5$, the bound
$4\sqrt{3g/2}+14+\eps$ derived there from
Heydt et al.~\cite{HeydtEtAl2025} is smaller.

\begin{corollary}
\label{cor:genus}
For every fixed integer $g\ge0$ and $\eps>0$, minimum dominating set
on graphs of Euler genus at most $g$ admits a deterministic
$(25+\eps)$-approximation in $C(g,\eps)$ rounds of $\LOCAL$.
\end{corollary}

\section*{Use of generative AI}
The structural idea was developed by the author. Generative AI assistants
were used to help check constants in the proofs, check references, find
and check examples, edit and organize the English text, and prepare LaTeX code.
All figures were drawn in TikZ with help from generative AI.
The author checked all mathematical proofs and reviewed each figure,
and takes full responsibility for the results and the final manuscript.

\bibliographystyle{plainurl}
\bibliography{references}

\begin{thebibliography}{10}

\bibitem{AmiriEtAl2018Expansion}
Saeed~Akhoondian Amiri, Patrice~Ossona de~Mendez, Roman Rabinovich, and
  Sebastian Siebertz.
\newblock Distributed domination on graph classes of bounded expansion.
\newblock In {\em Proceedings of the 30th ACM Symposium on Parallelism in
  Algorithms and Architectures (SPAA)}, pages 143--151, 2018.
\newblock \href {https://doi.org/10.1145/3210377.3210383}
  {\path{doi:10.1145/3210377.3210383}}.

\bibitem{AmiriSchmidSiebertz2016}
Saeed~Akhoondian Amiri, Stefan Schmid, and Sebastian Siebertz.
\newblock A local constant factor {MDS} approximation for bounded genus graphs.
\newblock In {\em Proceedings of the 2016 ACM Symposium on Principles of
  Distributed Computing}, pages 227--233. ACM, 2016.
\newblock \href {https://doi.org/10.1145/2933057.2933084}
  {\path{doi:10.1145/2933057.2933084}}.

\bibitem{AmiriSchmidSiebertz2019}
Saeed~Akhoondian Amiri, Stefan Schmid, and Sebastian Siebertz.
\newblock Distributed dominating set approximations beyond planar graphs.
\newblock {\em ACM Transactions on Algorithms}, 15(3):39:1--39:18, 2019.
\newblock \href {https://arxiv.org/abs/1705.09617} {\path{arXiv:1705.09617}},
  \href {https://doi.org/10.1145/3326170} {\path{doi:10.1145/3326170}}.

\bibitem{BonamyEtAl2021Outerplanar}
Marthe Bonamy, Linda Cook, Carla Groenland, and Alexandra Wesolek.
\newblock A tight local algorithm for the minimum dominating set problem in
  outerplanar graphs.
\newblock In {\em 35th International Symposium on Distributed Computing (DISC
  2021)}, volume 209 of {\em Leibniz International Proceedings in Informatics
  (LIPIcs)}, pages 13:1--13:18, 2021.
\newblock \href {https://doi.org/10.4230/LIPIcs.DISC.2021.13}
  {\path{doi:10.4230/LIPIcs.DISC.2021.13}}.

\bibitem{BonamyEtAl2026}
Marthe Bonamy, Avinandan Das, Cyril Gavoille, Timoth{\'e} Picavet, Jukka
  Suomela, and Alexandra Wesolek.
\newblock Meta-theorems for cuttable distributed problems.
\newblock In {\em Proceedings of the 45th ACM Symposium on Principles of
  Distributed Computing (PODC)}, pages 382--389, 2026.
\newblock Full version with appendices: arXiv:2605.19157v1.
\newblock URL: \url{https://arxiv.org/abs/2605.19157v1}, \href
  {https://arxiv.org/abs/2605.19157} {\path{arXiv:2605.19157}}, \href
  {https://doi.org/10.1145/3796701.3815958}
  {\path{doi:10.1145/3796701.3815958}}.

\bibitem{CzygrinowHanckowiakWawrzyniak2008}
Andrzej Czygrinow, Micha{\l} Ha{\'n}{\'c}kowiak, and Wojciech Wawrzyniak.
\newblock Fast distributed approximations in planar graphs.
\newblock In Gadi Taubenfeld, editor, {\em Distributed Computing (DISC 2008)},
  volume 5218 of {\em Lecture Notes in Computer Science}, pages 78--92.
  Springer, 2008.
\newblock \href {https://doi.org/10.1007/978-3-540-87779-0_6}
  {\path{doi:10.1007/978-3-540-87779-0_6}}.

\bibitem{CzygrinowEtAl2018}
Andrzej Czygrinow, Micha{\l} Ha{\'n}{\'c}kowiak, Wojciech Wawrzyniak, and
  Marcin Witkowski.
\newblock Distributed approximation algorithms for the minimum dominating set
  in {$K_h$}-minor-free graphs.
\newblock In {\em 29th International Symposium on Algorithms and Computation
  (ISAAC 2018)}, volume 123 of {\em Leibniz International Proceedings in
  Informatics (LIPIcs)}, pages 22:1--22:12, 2018.
\newblock \href {https://doi.org/10.4230/LIPIcs.ISAAC.2018.22}
  {\path{doi:10.4230/LIPIcs.ISAAC.2018.22}}.

\bibitem{CzygrinowEtAl2019Broadcast}
Andrzej Czygrinow, Micha{\l} Ha{\'n}{\'c}kowiak, Wojciech Wawrzyniak, and
  Marcin Witkowski.
\newblock Distributed {CONGEST\_BC} constant approximation of {MDS} in bounded
  genus graphs.
\newblock {\em Theoretical Computer Science}, 757:1--10, 2019.
\newblock \href {https://doi.org/10.1016/j.tcs.2018.07.008}
  {\path{doi:10.1016/j.tcs.2018.07.008}}.

\bibitem{CzygrinowHanckowiakWitkowski2022}
Andrzej Czygrinow, Micha{\l} Han{\'c}kowiak, and Marcin Witkowski.
\newblock Distributed distance domination in graphs with no {$K_{2,t}$}-minor.
\newblock {\em Theoretical Computer Science}, 916:22--30, 2022.
\newblock \href {https://arxiv.org/abs/2203.03229} {\path{arXiv:2203.03229}},
  \href {https://doi.org/10.1016/j.tcs.2022.03.001}
  {\path{doi:10.1016/j.tcs.2022.03.001}}.

\bibitem{DeurerKuhnMaus2019}
Janosch Deurer, Fabian Kuhn, and Yannic Maus.
\newblock Deterministic distributed dominating set approximation in the
  {CONGEST} model.
\newblock In {\em Proceedings of the 2019 ACM Symposium on Principles of
  Distributed Computing (PODC)}, pages 94--103, 2019.
\newblock \href {https://doi.org/10.1145/3293611.3331626}
  {\path{doi:10.1145/3293611.3331626}}.

\bibitem{Diestel2025}
Reinhard Diestel.
\newblock {\em Graph Theory}, volume 173 of {\em Graduate Texts in
  Mathematics}.
\newblock Springer, 6 edition, 2025.
\newblock \href {https://doi.org/10.1007/978-3-662-70107-2}
  {\path{doi:10.1007/978-3-662-70107-2}}.

\bibitem{HeydtEtAl2025}
Ozan Heydt, Simeon Kublenz, Patrice~Ossona de~Mendez, Sebastian Siebertz, and
  Alexandre Vigny.
\newblock Distributed domination on sparse graph classes.
\newblock {\em European Journal of Combinatorics}, 123:103773, 2025.
\newblock Numbered references in this paper follow the full version,
  arXiv:2207.02669v1.
\newblock URL: \url{https://arxiv.org/abs/2207.02669v1}, \href
  {https://arxiv.org/abs/2207.02669} {\path{arXiv:2207.02669}}, \href
  {https://doi.org/10.1016/j.ejc.2023.103773}
  {\path{doi:10.1016/j.ejc.2023.103773}}.

\bibitem{HeydtSiebertzVigny2022}
Ozan Heydt, Sebastian Siebertz, and Alexandre Vigny.
\newblock Local planar domination revisited.
\newblock In {\em Structural Information and Communication Complexity (SIROCCO
  2022)}, volume 13298 of {\em Lecture Notes in Computer Science}, pages
  154--173, 2022.
\newblock \href {https://arxiv.org/abs/2111.14506} {\path{arXiv:2111.14506}},
  \href {https://doi.org/10.1007/978-3-031-09993-9_9}
  {\path{doi:10.1007/978-3-031-09993-9_9}}.

\bibitem{HilkeLenzenSuomela2014}
Miikka Hilke, Christoph Lenzen, and Jukka Suomela.
\newblock Brief announcement: Local approximability of minimum dominating set
  on planar graphs.
\newblock In {\em Proceedings of the 2014 ACM Symposium on Principles of
  Distributed Computing (PODC)}, pages 344--346, 2014.
\newblock \href {https://arxiv.org/abs/1402.2549} {\path{arXiv:1402.2549}},
  \href {https://doi.org/10.1145/2611462.2611504}
  {\path{doi:10.1145/2611462.2611504}}.

\bibitem{KublenzSiebertzVigny2021}
Simeon Kublenz, Sebastian Siebertz, and Alexandre Vigny.
\newblock Constant round distributed domination on graph classes with bounded
  expansion.
\newblock In {\em Structural Information and Communication Complexity (SIROCCO
  2021)}, volume 12810 of {\em Lecture Notes in Computer Science}, pages
  334--351, 2021.
\newblock Corrected full version: arXiv:2012.02701v3.
\newblock URL: \url{https://arxiv.org/abs/2012.02701v3}, \href
  {https://arxiv.org/abs/2012.02701} {\path{arXiv:2012.02701}}, \href
  {https://doi.org/10.1007/978-3-030-79527-6_19}
  {\path{doi:10.1007/978-3-030-79527-6_19}}.

\bibitem{KuhnMoscibrodaWattenhofer2016}
Fabian Kuhn, Thomas Moscibroda, and Roger Wattenhofer.
\newblock Local computation: Lower and upper bounds.
\newblock {\em Journal of the ACM}, 63(2):17:1--17:44, 2016.
\newblock \href {https://doi.org/10.1145/2742012} {\path{doi:10.1145/2742012}}.

\bibitem{LenzenPignoletWattenhofer2013}
Christoph Lenzen, Yvonne-Anne Pignolet, and Roger Wattenhofer.
\newblock Distributed minimum dominating set approximations in restricted
  families of graphs.
\newblock {\em Distributed Computing}, 26(2):119--137, 2013.
\newblock \href {https://doi.org/10.1007/s00446-013-0186-z}
  {\path{doi:10.1007/s00446-013-0186-z}}.

\bibitem{Linial1992}
Nathan Linial.
\newblock Locality in distributed graph algorithms.
\newblock {\em SIAM Journal on Computing}, 21(1):193--201, 1992.
\newblock \href {https://doi.org/10.1137/0221015} {\path{doi:10.1137/0221015}}.

\bibitem{NaorStockmeyer1995}
Moni Naor and Larry~J. Stockmeyer.
\newblock What can be computed locally?
\newblock {\em SIAM Journal on Computing}, 24(6):1259--1277, 1995.
\newblock \href {https://doi.org/10.1137/S0097539793254571}
  {\path{doi:10.1137/S0097539793254571}}.

\bibitem{Suomela2013Survey}
Jukka Suomela.
\newblock Survey of local algorithms.
\newblock {\em ACM Computing Surveys}, 45(2):24:1--24:40, 2013.
\newblock \href {https://doi.org/10.1145/2431211.2431223}
  {\path{doi:10.1145/2431211.2431223}}.

\bibitem{Wawrzyniak2014}
Wojciech Wawrzyniak.
\newblock A strengthened analysis of a local algorithm for the minimum
  dominating set problem in planar graphs.
\newblock {\em Information Processing Letters}, 114(3):94--98, 2014.
\newblock \href {https://doi.org/10.1016/j.ipl.2013.11.008}
  {\path{doi:10.1016/j.ipl.2013.11.008}}.

\bibitem{Wawrzyniak2015Anonymous}
Wojciech Wawrzyniak.
\newblock A local approximation algorithm for minimum dominating set problem in
  anonymous planar networks.
\newblock {\em Distributed Computing}, 28(5):321--331, 2015.
\newblock \href {https://doi.org/10.1007/s00446-015-0247-6}
  {\path{doi:10.1007/s00446-015-0247-6}}.

\bibitem{Wawrzyniak2025TriangleFree}
Wojciech Wawrzyniak.
\newblock A local 6-approximation distributed algorithm for minimum dominating
  set problem in planar triangle-free graphs.
\newblock {\em Algorithms}, 18(5):280, 2025.
\newblock \href {https://doi.org/10.3390/a18050280}
  {\path{doi:10.3390/a18050280}}.

\end{thebibliography}
\appendix

\section{Proof of the partner bound}
\label{app:partners-proof}

We give the argument of Heydt et al.~\cite[Lemmas 8.1 and 8.2]{HeydtEtAl2025}
for cover size six and partner threshold nineteen.

\begin{restatedpartners}
If $Z\subseteq V\setminus\{v\}$, $|Z|\le6$, and $N(v)\subseteq N[Z]$,
then $P_v\subseteq Z$.
Moreover $D_1\cap D_2=\varnothing$, and every vertex of $D_2$ has at most
six partners.
\end{restatedpartners}
\begin{proof}
For $z\in P_v$, take nineteen common neighbors of $v,z$ in $R_1$.
Assign each to an element of $Z$ whose closed neighborhood contains it.
Some $c\in Z$ covers at least four, since $19>6\cdot3$.
If $z\notin Z$, the vertices $v,z,c$ are distinct. Remove $c$ itself
from these four vertices if necessary. At least three remain, each
adjacent to $v,z,c$. Common open neighbors cannot equal $v$ or $z$,
and now cannot equal $c$ either. These six vertices contain a
$K_{3,3}$ subgraph, contradicting planarity. Hence $z\in Z$.

If $v\in D_1$, its entire neighborhood is dominated by $D_1$, so
$N_1(v)=\varnothing$ and $P_v=\varnothing$.
Thus $v\in D_2$ implies $v\notin D_1$, so $N(v)$ can be dominated
by at most six other vertices.
The first assertion then gives $|P_v|\le6$.
\end{proof}

\section{Proof of the residual-degree bound}
\label{app:residual-degree-proof}

We give the argument of Heydt et al.~\cite[Lemma 8.7]{HeydtEtAl2025}
for cover size six and partner threshold nineteen.

\begin{restatedresidualdegree}
The set $R=V\setminus N[S]$ consists of the vertices still undominated
after the first two phases, where $S=D_1\cup D_2$.
For every vertex $v\in V(G)$,
\[
 |N(v)\cap R|\le114.
\]
\end{restatedresidualdegree}
\begin{proof}
Suppose $v$ has at least 115 neighbors in $R$. Then $v\notin S$,
since otherwise those neighbors would be dominated by $S$.
In particular $v\notin D_1$, so choose $Z\subseteq V\setminus\{v\}$
with $|Z|\le6$ and $N(v)\subseteq N[Z]$.
Some $c\in Z$ covers at least twenty of the 115 neighbors, since
$115>6\cdot19$. Removing $c$ itself if necessary leaves at least
nineteen common open neighbors of $v,c$ in $R\subseteq R_1$.
Thus $c\in P_v$ and $v\in D_2\subseteq S$, a contradiction.
\end{proof}

\section{Sharpness and the meaning of the contact bound}
\label{sec:sharpness}

Counting the edges of the contact graph $H$ alone does not give the bound
on $W$. Several vertices of $G$ can have contacts represented by the same
edge of $H$. We now construct examples where the bound on $W$ holds with
equality and the coefficient one in the bound on $|A|$ cannot be reduced.

\begin{proposition}
\label{prop:sharp}
For arbitrarily large $|D|$, there are planar graphs with a dominating
set $D$ and an owner assignment such that
\[
 W=4|D|-12,\qquad |A|=|D|-6.
\]
In these examples, $D$ is a minimum dominating set. Thus the coefficient
one in the bound on $|A|$ cannot be replaced by a smaller constant that
works for all graphs, even when $D$ is minimum.
\end{proposition}
\begin{proof}
For each integer $r\ge2$, we draw $r$ nested triangles.
We join corresponding vertices of consecutive triangles and add one
diagonal in each region with four sides,
in the same cyclic direction, as in \cref{fig:sharpness-construction}(a).
The inner and outer faces are also triangles, so the resulting graph
$T$ is a plane triangulation with $q=3r$ vertices.
The six vertices on the innermost and outermost triangles have degree
four. Every other vertex has two neighbors on its own triangle and two
on each adjacent triangle, so its degree is six.

For each vertex $v$ of $T$, we add a new leaf $d_v$ adjacent only to $v$.
This gives a planar graph $G$.
We set $D=\{d_v:v\in V(T)\}$ and assign $v$ to $d_v$.
Then $B_{d_v}=\{v,d_v\}$ and $k_v=\deg_T(v)$.
Figure~\ref{fig:sharpness-construction} shows the construction for $r=4$.

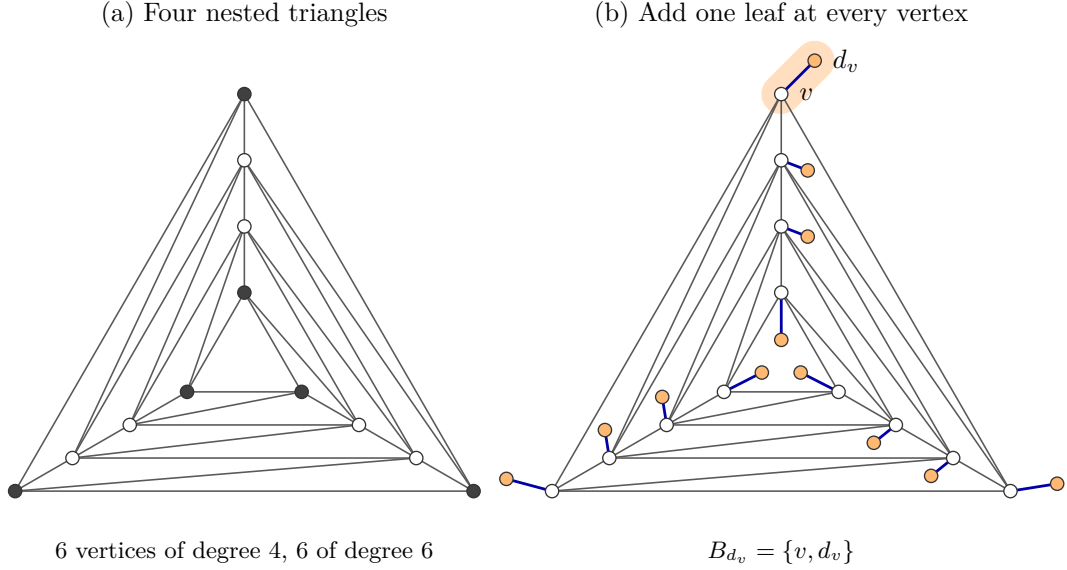
\begin{figure}[!ht]
\centering
\begin{tikzpicture}[every node/.style={font=\small},
  edge/.style={draw=black!65,line width=.6pt},
  owner/.style={draw=blue!65!black,line width=1.05pt}]
\begin{scope}[xshift=0cm]
\coordinate (v0) at (0.00000,3.50000);
\coordinate (v1) at (-3.03109,-1.75000);
\coordinate (v2) at (3.03109,-1.75000);
\coordinate (v3) at (0.00000,2.62500);
\coordinate (v4) at (-2.27332,-1.31250);
\coordinate (v5) at (2.27332,-1.31250);
\coordinate (v6) at (0.00000,1.75000);
\coordinate (v7) at (-1.51554,-0.87500);
\coordinate (v8) at (1.51554,-0.87500);
\coordinate (v9) at (0.00000,0.87500);
\coordinate (v10) at (-0.75777,-0.43750);
\coordinate (v11) at (0.75777,-0.43750);
\draw[edge] (v0)--(v1);
\draw[edge] (v0)--(v2);
\draw[edge] (v0)--(v3);
\draw[edge] (v0)--(v4);
\draw[edge] (v1)--(v2);
\draw[edge] (v1)--(v4);
\draw[edge] (v1)--(v5);
\draw[edge] (v2)--(v3);
\draw[edge] (v2)--(v5);
\draw[edge] (v3)--(v4);
\draw[edge] (v3)--(v5);
\draw[edge] (v3)--(v6);
\draw[edge] (v3)--(v7);
\draw[edge] (v4)--(v5);
\draw[edge] (v4)--(v7);
\draw[edge] (v4)--(v8);
\draw[edge] (v5)--(v6);
\draw[edge] (v5)--(v8);
\draw[edge] (v6)--(v7);
\draw[edge] (v6)--(v8);
\draw[edge] (v6)--(v9);
\draw[edge] (v6)--(v10);
\draw[edge] (v7)--(v8);
\draw[edge] (v7)--(v10);
\draw[edge] (v7)--(v11);
\draw[edge] (v8)--(v9);
\draw[edge] (v8)--(v11);
\draw[edge] (v9)--(v10);
\draw[edge] (v9)--(v11);
\draw[edge] (v10)--(v11);
\filldraw[fill=black!75,draw=black!80,line width=.5pt] (v0) circle (2.5pt);
\filldraw[fill=black!75,draw=black!80,line width=.5pt] (v1) circle (2.5pt);
\filldraw[fill=black!75,draw=black!80,line width=.5pt] (v2) circle (2.5pt);
\filldraw[fill=white,draw=black!80,line width=.5pt] (v3) circle (2.5pt);
\filldraw[fill=white,draw=black!80,line width=.5pt] (v4) circle (2.5pt);
\filldraw[fill=white,draw=black!80,line width=.5pt] (v5) circle (2.5pt);
\filldraw[fill=white,draw=black!80,line width=.5pt] (v6) circle (2.5pt);
\filldraw[fill=white,draw=black!80,line width=.5pt] (v7) circle (2.5pt);
\filldraw[fill=white,draw=black!80,line width=.5pt] (v8) circle (2.5pt);
\filldraw[fill=black!75,draw=black!80,line width=.5pt] (v9) circle (2.5pt);
\filldraw[fill=black!75,draw=black!80,line width=.5pt] (v10) circle (2.5pt);
\filldraw[fill=black!75,draw=black!80,line width=.5pt] (v11) circle (2.5pt);
\node at (0,4.55) {(a) Four nested triangles};
\node[font=\footnotesize] at (0,-2.55) {$6$ vertices of degree $4$, $6$ of degree $6$};
\end{scope}
\begin{scope}[xshift=7.1cm]
\coordinate (v0) at (0.00000,3.50000);
\coordinate (v1) at (-3.03109,-1.75000);
\coordinate (v2) at (3.03109,-1.75000);
\coordinate (v3) at (0.00000,2.62500);
\coordinate (v4) at (-2.27332,-1.31250);
\coordinate (v5) at (2.27332,-1.31250);
\coordinate (v6) at (0.00000,1.75000);
\coordinate (v7) at (-1.51554,-0.87500);
\coordinate (v8) at (1.51554,-0.87500);
\coordinate (v9) at (0.00000,0.87500);
\coordinate (v10) at (-0.75777,-0.43750);
\coordinate (v11) at (0.75777,-0.43750);
\coordinate (v12) at (0.44194,3.94194);
\coordinate (v13) at (-3.63479,-1.58824);
\coordinate (v14) at (3.64839,-1.65223);
\coordinate (v15) at (0.35009,2.49061);
\coordinate (v16) at (-2.33198,-0.94212);
\coordinate (v17) at (1.98189,-1.54850);
\coordinate (v18) at (0.35009,1.61561);
\coordinate (v19) at (-1.57421,-0.50462);
\coordinate (v20) at (1.22411,-1.11100);
\coordinate (v21) at (-0.00000,0.25000);
\coordinate (v22) at (-0.25658,-0.18213);
\coordinate (v23) at (0.25658,-0.18213);
\draw[draw=orange!25,line width=15pt,line cap=round] (v0)--(v12);
\draw[edge] (v0)--(v1);
\draw[edge] (v0)--(v2);
\draw[edge] (v0)--(v3);
\draw[edge] (v0)--(v4);
\draw[edge] (v1)--(v2);
\draw[edge] (v1)--(v4);
\draw[edge] (v1)--(v5);
\draw[edge] (v2)--(v3);
\draw[edge] (v2)--(v5);
\draw[edge] (v3)--(v4);
\draw[edge] (v3)--(v5);
\draw[edge] (v3)--(v6);
\draw[edge] (v3)--(v7);
\draw[edge] (v4)--(v5);
\draw[edge] (v4)--(v7);
\draw[edge] (v4)--(v8);
\draw[edge] (v5)--(v6);
\draw[edge] (v5)--(v8);
\draw[edge] (v6)--(v7);
\draw[edge] (v6)--(v8);
\draw[edge] (v6)--(v9);
\draw[edge] (v6)--(v10);
\draw[edge] (v7)--(v8);
\draw[edge] (v7)--(v10);
\draw[edge] (v7)--(v11);
\draw[edge] (v8)--(v9);
\draw[edge] (v8)--(v11);
\draw[edge] (v9)--(v10);
\draw[edge] (v9)--(v11);
\draw[edge] (v10)--(v11);
\draw[owner] (v0)--(v12);
\draw[owner] (v1)--(v13);
\draw[owner] (v2)--(v14);
\draw[owner] (v3)--(v15);
\draw[owner] (v4)--(v16);
\draw[owner] (v5)--(v17);
\draw[owner] (v6)--(v18);
\draw[owner] (v7)--(v19);
\draw[owner] (v8)--(v20);
\draw[owner] (v9)--(v21);
\draw[owner] (v10)--(v22);
\draw[owner] (v11)--(v23);
\filldraw[fill=white,draw=black!80,line width=.5pt] (v0) circle (2.5pt);
\filldraw[fill=white,draw=black!80,line width=.5pt] (v1) circle (2.5pt);
\filldraw[fill=white,draw=black!80,line width=.5pt] (v2) circle (2.5pt);
\filldraw[fill=white,draw=black!80,line width=.5pt] (v3) circle (2.5pt);
\filldraw[fill=white,draw=black!80,line width=.5pt] (v4) circle (2.5pt);
\filldraw[fill=white,draw=black!80,line width=.5pt] (v5) circle (2.5pt);
\filldraw[fill=white,draw=black!80,line width=.5pt] (v6) circle (2.5pt);
\filldraw[fill=white,draw=black!80,line width=.5pt] (v7) circle (2.5pt);
\filldraw[fill=white,draw=black!80,line width=.5pt] (v8) circle (2.5pt);
\filldraw[fill=white,draw=black!80,line width=.5pt] (v9) circle (2.5pt);
\filldraw[fill=white,draw=black!80,line width=.5pt] (v10) circle (2.5pt);
\filldraw[fill=white,draw=black!80,line width=.5pt] (v11) circle (2.5pt);
\filldraw[fill=orange!55,draw=black!80,line width=.5pt] (v12) circle (2.5pt);
\filldraw[fill=orange!55,draw=black!80,line width=.5pt] (v13) circle (2.5pt);
\filldraw[fill=orange!55,draw=black!80,line width=.5pt] (v14) circle (2.5pt);
\filldraw[fill=orange!55,draw=black!80,line width=.5pt] (v15) circle (2.5pt);
\filldraw[fill=orange!55,draw=black!80,line width=.5pt] (v16) circle (2.5pt);
\filldraw[fill=orange!55,draw=black!80,line width=.5pt] (v17) circle (2.5pt);
\filldraw[fill=orange!55,draw=black!80,line width=.5pt] (v18) circle (2.5pt);
\filldraw[fill=orange!55,draw=black!80,line width=.5pt] (v19) circle (2.5pt);
\filldraw[fill=orange!55,draw=black!80,line width=.5pt] (v20) circle (2.5pt);
\filldraw[fill=orange!55,draw=black!80,line width=.5pt] (v21) circle (2.5pt);
\filldraw[fill=orange!55,draw=black!80,line width=.5pt] (v22) circle (2.5pt);
\filldraw[fill=orange!55,draw=black!80,line width=.5pt] (v23) circle (2.5pt);
\node[right=3pt] at (v0) {$v$};
\node[right=3pt] at (v12) {$d_v$};
\node at (0,4.55) {(b) Add one leaf at every vertex};
\node[font=\footnotesize] at (0,-2.55) {$B_{d_v}=\{v,d_v\}$};
\end{scope}
\end{tikzpicture}
\caption{The construction with four nested triangles.
(a) The triangulation $T$. Black vertices have degree four and white
vertices have degree six.
(b) The graph $G$ after adding one leaf at each vertex. Orange leaves
form $D$, and the shaded pair is one owner block $B_{d_v}=\{v,d_v\}$.}
\label{fig:sharpness-construction}
\end{figure}

The set $A$ consists of the $q-6$ vertices of degree six in $T$.
Each has weight four, and the other six non-centers have weight two.
Thus
\[
 |A|=q-6,\qquad W=4(q-6)+2\cdot6=4q-12.
\]
The closed neighborhoods $N[d_v]=\{v,d_v\}$ are pairwise disjoint.
Every dominating set must contain a vertex from each of these sets
to dominate all leaves. It therefore has at least $q$ vertices.
Our set $D$ dominates $G$ and has size $q$, so it is minimum.

These examples also explain the coefficient in the bound used in
\cref{sec:algorithm}. For a vertex $v$ of degree six in $T$, the set
$N_G(v)$ contains $d_v$ and the six neighbors of $v$ in $T$.
Among centers in $D$, only $d_v$ dominates $d_v$, and only $d_w$
dominates each neighbor $w$ of $v$ in $T$.
Thus covering $N_G(v)$ with centers from $D$ requires seven centers.
This holds for all $q-6$ vertices in $A$, and $(q-6)/q$ tends to one
as $r$ grows. Hence the coefficient one cannot be reduced.
Here the covering set must consist of centers in $D$, not arbitrary
vertices of the graph.
\end{proof}

These examples show that the coefficients in the bounds cannot be
reduced. They do not show that the final approximation factor eight is
the best possible for an algorithm.

\end{document}